\documentclass[ASNA,twocolumn]{USG} 
\usepackage{anyfontsize} %

\usepackage{colortbl}   
\usepackage{xcolor}     
\usepackage{cuted}
\usepackage{steinmetz}
\usepackage{mathtools}
\graphicspath{{./images/}}
\specialissue{Advances in Quantum Control, Robust Control and Negative Imaginary Systems}

\newtheorem{example}{Example}[section]

\newcommand{\bthm}{\begin{theorem}}
\newcommand{\ethm}{\end{theorem}}
\newcommand{\blem}{\begin{lemma}}
\newcommand{\elem}{\end{lemma}}
\newcommand{\bex}{\begin{example}}
\newcommand{\eex}{\end{example}}
\newcommand{\bprop}{\begin{proposition}}
\newcommand{\eprop}{\end{proposition}}
\newcommand{\bplm}{\begin{problem}}
\newcommand{\eplm}{\end{problem}}
\newcommand{\bmrk}{\begin{remark}}
\newcommand{\emrk}{\end{remark}}
\newcommand{\bdfn}{\begin{definition}}
\newcommand{\edfn}{\end{definition}}
\newcommand{\bcor}{\begin{corollary}}
\newcommand{\ecor}{\end{corollary}}

\newcommand{\beq}{\begin{equation}}
\newcommand{\eeq}{\end{equation}}
\newcommand{\beqm}{\begin{equation*}}
\newcommand{\eeqm}{\end{equation*}}
\newcommand{\beqn}{\begin{eqnarray}}
\newcommand{\eeqn}{\end{eqnarray}}
\newcommand{\beqnm}{\begin{eqnarray*}}
\newcommand{\eeqnm}{\end{eqnarray*}}
\newcommand{\bea}{\begin{align}}
\newcommand{\eea}{\end{align}}
\newcommand{\bead}{\begin{aligned}}
\newcommand{\eead}{\end{aligned}}
\newcommand{\beam}{\begin{align*}}
\newcommand{\eeam}{\end{align*}}
\newcommand{\bs}{\begin{subequations}}
\newcommand{\es}{\end{subequations}}

\newcommand{\bei}{\begin{itemize}}
\newcommand{\eei}{\end{itemize}}
\newcommand{\bed}{\begin{description}}
\newcommand{\eed}{\end{description}}
\newcommand{\bee}{\begin{enumerate}}
\newcommand{\eee}{\end{enumerate}}

\newcommand{\bey}{\begin{array}}
\newcommand{\eey}{\end{array}}

\newcommand{\beb}{\begin{thebibliography}{99}}
\newcommand{\eeb}{\end{thebibliography}}

\newcommand{\mbf}{\mathbf}

\newcommand{\la}{\label}

\articletype{RESEARCH ARTICLE}%

\journal{International Journal of Robust and Nonlinear Control}
\volume{0}
\copyyear{2026}
\startpage{1}
\articledoi{}

\begin{document}

\title{On Realization of Back-Action-Evading Measurements and Quantum Non-Demolition Variables via Linear Systems Engineering}
\transtitle{On Realization of Back-Action-Evading Measurements and Quantum Non-Demolition Variables via Linear Systems Engineering}
\author[1]{Zhiyuan Dong}[https://orcid.org/0000-0002-8762-7943]
\author[2]{Weichao Liang}[https://orcid.org/0000-0002-2467-0469]
\author[3,4]{Guofeng Zhang}[https://orcid.org/0000-0001-5854-5247]

\authormark{Zhiyuan Dong \textsc{et al.}}
\titlemark{On Realization of Back-Action-Evading Measurements and Quantum Non-Demolition Variables via Linear Systems Engineering}

\address[1]{\orgdiv{School of Science, College of Frontier Sciences, }\orgname{Harbin Institute of Technology, Shenzhen, }%
\orgaddress{\state{Guangdong, }\country{China}}}

\address[2]{\orgdiv{School of Automation Science and Engineering, Faculty of Electronic and Information Engineering, }\orgname{Xi'an Jiaotong University, }%
\orgaddress{\state{Shaanxi, }\country{China}}}

\address[3]{\orgdiv{Department of Applied Mathematics, }\orgname{The Hong Kong Polytechnic University, }%
\orgaddress{\state{Hong Kong SAR, }\country{China}}}

\address[4]{\orgdiv{Research Institute for Quantum Technology, }\orgname{The Hong Kong Polytechnic University, }%
\orgaddress{\state{Hong Kong SAR, }\country{China}}}


\corres{Guofeng Zhang  (\email{guofeng.zhang@polyu.edu.hk})}


\presentaddress{This is sample for present address text this is sample for present address text.}



\transkeywords{linear quantum systems | back-action-evading measurements | quantum non-demolition variables | quantum coherent feedback control}

\abstract[ABSTRACT]{We establish a framework for realizing back-action-evading (BAE) measurements and quantum non-demolition (QND) variables in linear quantum systems. The key condition, a purely imaginary Hamiltonian with a real or imaginary coupling operator, enables BAE measurements of conjugate observables. Symmetric coupling further yields QND variables. For non-compliant systems, coherent feedback is designed to engineer BAE measurements. Crucially, the QND interaction condition simultaneously ensures BAE measurements and promotes the coupling operator to a QND observable.}

 \transabstract[transABSTRACT]{We establish a framework for realizing back-action-evading (BAE) measurements and quantum non-demolition (QND) variables in linear quantum systems. The key condition, a purely imaginary Hamiltonian with a real or imaginary coupling operator, enables BAE measurements of conjugate observables. Symmetric coupling further yields QND variables. For non-compliant systems, coherent feedback is designed to engineer BAE measurements. Crucially, the QND interaction condition simultaneously ensures BAE measurements and promotes the coupling operator to a QND observable.}




\copyright{@ 2026 John Wiley \& Sons Ltd.}


\maketitle


\section{Introduction} \label{sec:intro}

In a typical indirect measurement within a quantum system \cite{belavkin1989nondemolition,Belavkin1994,HMW95,BvHJ07}, an auxiliary quantum system, often termed a probe, is deliberately coupled to the system of interest. Information about the latter is inferred by performing a direct measurement on the probe after their interaction. In this framework, the probe serves as an input channel before the interaction and transforms into an output channel afterwards. A fundamental consequence of this information acquisition process is the introduction of measurement back-action, which inevitably disturbs the dynamics of the system of interest.

A key objective in quantum control and metrology is to identify and protect parts of a quantum system that are inherently immune to such measurement-induced disturbances. If a specific subsystem of a linear quantum system is inaccessible to the probe—meaning it is neither influenced by the input probe nor detectable in the output—it is termed a decoherence-free subsystem (DFS) \cite{NY13,PDP17,ZGPG18,ZPL20}. Importantly, a DFS is not an isolated system; it interacts with other parts of the composite system; \cite[Fig. 2]{ZGPG18}. Its decoherence-free property implies that its state evolves unitarily and remains separable from the rest of the system's degrees of freedom, provided the system-environment coupling respects a certain symmetry. In the Heisenberg picture, this concept is closely related to the decoherence-free subspaces in the Schrödinger picture \cite[Sec. III-C]{DZWW23}, which play a crucial role in protecting quantum information for quantum computation \cite{TV08}.

While a DFS represents a structural immunity where certain system degrees of freedom are statically inaccessible to the probe, a more dynamical form of back-action evasion can be engineered at the level of input-output relationships. Specifically, a quantum back-action evading (BAE) measurement is achieved when a particular observable of the output probe, denoted $\boldsymbol{y}_q$, is insensitive to a specific conjugate observable of the input probe, denoted $\boldsymbol{u}_p$. In the language of linear systems and the Kalman canonical form, this condition corresponds to a zero transfer function from the input $\boldsymbol{u}_p$ to the output $\boldsymbol{y}_q$. This frequency-domain characterization is consistent with the standard linear systems approach to quantum measurement theory, where transfer functions have been rigorously established as a valid tool for describing BAE conditions \cite[Fig. 2]{NY14}. In continuous-variable BAE experiments, the quantity of interest is typically the measured quadrature of the output signal \cite{Liu_2022}; a zero transfer function in the frequency domain guarantees that the back-action noise from input $\boldsymbol{u}_p$ is evaded across all relevant frequencies. Quantum BAE measurements are therefore of paramount importance in quantum sensing and metrology \cite{WC13,BQND24}, as they enable precision beyond the standard quantum limit by selectively rendering the measurement immune to specific fluctuations of the input field.

The concept of back-action evasion focuses on the input-output relation—evading back-action from a specific input channel. A related but distinct paradigm concerns the system observables themselves: the notion of quantum non-demolition (QND) measurements \cite{BK96}. According to the Heisenberg uncertainty principle, the precision of simultaneous measurements of canonically conjugate observables (e.g., position and momentum, or amplitude and phase) is fundamentally limited. Consequently, an indirect measurement of an observable $\mathcal{O}$ generally introduces back-action onto its conjugate counterpart. However, if the Heisenberg-picture evolution of $\mathcal{O}$ is a function only of itself and is independent of its conjugate observables, then the measurement of $\mathcal{O}$ does not perturb its own future evolution. This defines a QND variable: an observable $\mathcal{O}$ is a QND variable if it can be measured repeatedly, with the measurement back-action confined entirely to its conjugate variables, leaving the measured observable's future trajectory unaffected \cite{TC10,WC13,NY14,ZGPG18,LOW+21}. While back-action evasion ensures that a specific output is free from a specific input noise, QND measurement ensures that the observable of interest is free from the measurement back-action generated by its own measurement.

This paper presents a systematic framework for realizing BAE measurements and designing QND variables using linear systems engineering. By leveraging the structural properties of linear quantum systems, we unify the concepts of BAE and QND within a common state-space representation, providing both a theoretical foundation and practical design criteria for BAE measurements and QND variables in continuous-mode measurement scenarios.
Some  notation commonly used in this paper are listed as follows.

\textit{Notation}.
\begin{itemize}
\item $\imath =\sqrt{-1}$ is the imaginary unit. $I_{k}$ is the identity matrix and $0_{k}$ the zero matrix in $\mathbf{C}^{k \times k}$. $\delta_{ij}$ denotes the Kronecker delta;
i.e.,~$I_k=[\delta_{ij}]$. $\delta(t)$ is the Dirac delta function. ${\rm Re}(X)$ denotes the real part of the matrix $X$; while ${\rm Im}(X)$ is its imaginary part.

\item $x^{\ast}$ denotes the complex conjugate of a complex number $x$ or
the adjoint of an operator $x$. Clearly. $(xy)^\ast = y^\ast x^\ast$.  Given two operators $\bf{x}$ and $\bf{y}$, their commutator is defined to be $[\bf{x},\bf{y}] \triangleq \bf{x}\bf{y}-\bf{y}\bf{x}$.

\item For a matrix $X=[x_{ij}]$ with  number or operator entries,  $X^{\top}=[x_{ji}]$ is the matrix transpose. Denote $X^{\#}=[x_{ij}^{\ast}]$, and $X^{\dagger}=(X^{\#})^{\top}$. For a vector $x$, we define $\breve{x}\triangleq \bigl[
\begin{smallmatrix}
x \\
x^{\#}
\end{smallmatrix}
\bigr]$.

\item Given two \textit{column} vectors of operators $\bf{X}$ and $\bf{Y}$ of the same length, their commutator is defined as
\begin{equation*}
[\bf{X},\bf{Y}^\top] \triangleq ([\bf{X}_j,\bf{Y}_k] ) =\bf{X}\bf{Y}^\top- (\bf{Y}\bf{X}^\top)^\top.    
\end{equation*}
If $\bf{X}$ is a \textit{row} vector of operators of length $m$ and $\bf{Y}$ is a \textit{column} vector of operators of length $n$, their commutator is defined as
\begin{equation}\label{dec19-6}
[\mbf{X},\mbf{Y}] \triangleq \left(\begin{array}{@{}ccc@{}}                               [\mbf{x}_1,\mbf{y}_1] & \cdots & [\mbf{x}_m,\mbf{y}_1] \\
                               {\vdots} & \ddots & \vdots \\
                               {[\mbf{x}_1,\mbf{y}_n]} & \cdots & [\mbf{x}_m,\mbf{y}_n]
\end{array}
                           \right)_{n\times m}=(\mbf{X}^\top\mbf{Y}^\top)^\top-\mbf{Y}\mbf{X}.
\end{equation}

\item Let $J_{k} \triangleq \mathrm{diag}(I_k,-I_k)$. For a matrix $X\in
\mathbf{C}^{2k\times 2r}$, define its $\flat$-adjoint by $X^{\flat }
\triangleq J_{r}X^{\dagger}J_{k}$. The $\flat$-adjoint operation enjoys the following  properties:
\begin{equation*}
(x_1 A + x_2 B)^{\flat}=x_1^{*}
A^{\flat} + x_2^{*} B^{\flat}, \ \ (AB)^{\flat}=B^{\flat} A^{\flat}, \ \
(A^{\flat})^{\flat}=A,    
\end{equation*}
where $x_1,x_2\in \mathbf{C}$.

\item Given two matrices $U$, $V\in \mathbf{C}^{k\times r}$, define their
 \emph{doubled-up} \cite{GJN10} as  $\Delta
(U,V) \triangleq
\bigl[
\begin{smallmatrix}
U & V \\
V^{\#} & U^{\#}
\end{smallmatrix}
\bigr]$.  The set
of doubled-up matrices is closed under addition, multiplication and $\flat$ adjoint operation.


\item Let $\mathbb{J}_{k} \triangleq \bigl[
\begin{smallmatrix}
0_{k} & I_k \\
-I_k & 0_{k}
\end{smallmatrix}
\bigr]$. For a matrix $X\in \mathbf{C}^{2k\times 2r}$, define its $\sharp$-
\emph{adjoint} $X^{\sharp}$ by $X^{\sharp} \triangleq -\mathbb{J}_{r}X^{\dagger}
\mathbb{J}_{k}$. The $\sharp$-\emph{adjoint} satisfies properties similar to
the usual adjoint, namely
\begin{equation*}
(x_1 A + x_2 B)^{\sharp}=x_1^{*} A^{\sharp} + x_2^{*}
B^{\sharp}, \ \ (AB)^{\sharp}=B^{\sharp} A^{\sharp},  \ \ (A^{\sharp})^{
\sharp}=A.    
\end{equation*}

\end{itemize}

The remainder of this paper is organized as follows. 
Section \ref{preliminaries} briefly reviews the essential preliminaries of linear quantum systems, including the quantum stochastic differential equation, the Heisenberg-picture dynamics, and the associated input-output relations. 
In Section \ref{sec:qndbae}, we establish a systematic framework for characterizing BAE measurements. Specifically, we first derive sufficient conditions on the system parameters that guarantee bilateral BAE measurements, and then subsequently analyze unilateral BAE conditions under relaxed assumptions on the Hamiltonian and coupling operator.
Section \ref{QND INT} is devoted to a detailed investigation of the QND interaction. The analysis proceeds from the single-input-single-output (SISO) case to the more general multi-input-multi-output (MIMO) setting, employing annihilation-creation operator representation. The Kalman canonical form is also used to analyze the input-output structure associated with BAE measurements and QND variables. In Section \ref{synthesis}, we address the synthesis problem. A coherent feedback control scheme is proposed to actively engineer BAE measurements for systems that do not inherently satisfy the structural conditions derived earlier. In addition, a direct coupling approach for realizing QND variables is presented and illustrated using a quantum optomechanical system as a case study. Finally, Section \ref{Conclu} concludes the paper and provides a brief outlook on potential future research directions.

\section{Preliminaries}\label{preliminaries}

The time evolution of a linear quantum system is governed by a unitary operator $\mbf{U}(t,t_0)$, which satisfies the following quantum stochastic differential equation (QSDE)

\begin{equation}\label{dU}\begin{aligned}
d\mbf{U}(t,t_0)=&\bigg[\left(-\imath \mbf{H}-\frac{1}{2}\mbf{L}^\dagger \mbf{L}\right)dt-\mbf{L}^\dagger S d\mbf{B}_{\rm in}(t)+d\mbf{B}_{\rm in}^\dagger(t)\mbf{L} \\
&+{\rm Tr}\left[(S-I)d\mbf{\Lambda}^\top(t)\right]\bigg]\mbf{U}(t,t_0),  \end{aligned}\end{equation}
where $\mbf{U}(t_0,t_0)=I$. $S$ denotes the scattering matrix (a unitary matrix), $\mbf{L}=\left[\begin{array}{cc}
C_- & C_+ 
\end{array}\right]\breve{\mbf{a}}$ is the coupling operator, and $\mbf{H}$ is the system Hamiltonian of the form $\mbf{H}=\frac{1}{2}\breve{\mbf{a}}^\dagger\Omega\breve{\mbf{a}}$, where $\Omega=\Delta(\Omega_-,\Omega_+)$ is Hermitian. The integrated input annihilation, creation, and gauge processes are
\begin{equation*}\begin{aligned}
&\mbf{B}_{\rm in}(t)=\int_{t_0}^t \mbf{b}_{\rm in}(s)ds, ~~ \mbf{B}_{\rm in}^\#(t)=\int_{t_0}^t \mbf{b}_{\rm in}^\#(s)ds, \\
&\mbf{\Lambda}(t)=\int_{t_0}^t \mbf{b}_{\rm in}^\#(s) \mbf{b}_{\rm in}^\top(s)ds,
\end{aligned}\end{equation*}
respectively, and satisfy the following commutation relation
\begin{equation*}
\left[\mbf{b}_j(t),\mbf{b}^\ast_k(s)\right]=\delta_{jk}\delta(t-s).
\end{equation*}
In the Heisenberg picture, denote a system operator $\mbf{X}(t)=\mbf{U}^\ast(t,t_0)(\mbf{X}\otimes I_{\rm field})\mbf{U}(t,t_0)$, and its dynamical evolution can be described by a quantum stochastic differential equation
\begin{equation}\label{dX}\begin{aligned}
d\mbf{X}(t)=&\mathcal{L}(\mbf{X}(t))dt+{\rm Tr}\left[(S^\dagger \mbf{X}(t)S-\mbf{X}(t))d\mbf{\Lambda}^\top(t)\right] \\
&+[\mbf{L}^\dagger(t),\mbf{X}(t)]Sd\mbf{B}_{\rm in}(t)+d\mbf{B}_{\rm in}^\dagger(t) S^\dagger[\mbf{X}(t),\mbf{L}(t)],    
\end{aligned}\end{equation}
where the superoperator 
\begin{equation*}
\mathcal{L}(\mbf{X}(t))=-\imath[\mbf{X}(t),\mbf{H}(t)]+\frac{1}{2}\mbf{L}^\dagger(t)[\mbf{X}(t),\mbf{L}(t)]+\frac{1}{2}[\mbf{L}^\dagger(t),\mbf{X}(t)]\mbf{L}(t).    
\end{equation*}
The input-output relation is given by
\begin{equation}\label{eq:io}
d\mbf{B}_{\rm out}(t)=\mbf{L}(t)dt+Sd\mbf{B}_{\rm in}(t).
\end{equation}

In the Schr{\"o}dinger picture, the conditioned density operator $\rho_c$ over the joint system-field space is governed by the following quantum stochastic master equation (QSME)
\begin{equation}\label{Jan12-2}\begin{aligned}
d\rho_c(t)=&\mathcal{L}^\star(\rho_c(t))dt+\sum_{i=1}^m\big\{\mbf{L}_i\rho_c(t)+\rho_c(t)\mbf{L}_i^* \\
&-{\rm Tr}[\rho_c(t)(\mbf{L}_i+\mbf{L}_i^*)]\rho_c(t)\big\}d\nu_i(t),
\end{aligned}\end{equation}
where the system contains $m$ input-output channels, the superoperator $\mathcal{L}^\star(\rho_c(t))$ is defined as
\begin{equation*}\begin{aligned}
\mathcal{L}^\star(\rho_c(t))=&-\imath[\mbf{H},\rho_c(t)]+\sum_{i=1}^m\mbf{L}_i \rho_c(t) \mbf{L}_i^* \\
&-\frac{1}{2}\sum_{i=1}^m\mbf{L}_i^* \mbf{L}_i\rho_c(t)-\frac{1}{2}\sum_{i=1}^m\rho_c(t)\mbf{L}_i^* \mbf{L}_i,
\end{aligned}\end{equation*}
and the innovation process 
\begin{equation*}\begin{aligned}
d\nu(t)=&\left[d\nu_1(t), d\nu_2(t)\cdots,d\nu_m(t)\right]^\top \\
=&d\mbf{Q}_{\rm out}(t)-{\rm Tr}
[\rho_c(t)(\mbf{L}+\mbf{L}^\#)]dt,
\end{aligned}\end{equation*}
where $d\nu(t)d\nu^\top(t)=Idt$, and $\mbf{Q}_{\rm out}(t)=\frac{1}{\sqrt{2}}\left(\mbf{B}_{\rm out}(t)+\mbf{B}_{\rm out}^\#(t)\right)$ being assumed to be observed continuously.

\section{Structural characterization of BAE measurements}\label{sec:qndbae}

As introduced in Section \ref{sec:intro}, a quantum BAE measurement renders a chosen output observable insensitive to a specific conjugate input quadrature, which is mathematically equivalent to imposing a zero transfer function between that input–output pair, e.g., $\mathbb{G}_{\boldsymbol{u}_p\rightarrow\boldsymbol{y}_q}[s]=0$. This condition ensures that the back-action noise originating from the selected input quadrature cannot propagate through the system dynamics into the measurement record. In the following, we formalize the BAE condition within the linear quantum system framework using linear stochastic differential equations and transfer-function matrices. By expressing the input–output mapping in the state-space form, we translate the condition into algebraic constraints on the system  matrices, from which explicit criteria for the analysis and design and of quantum BAE measurements are derived.

A linear quantum system with the triple $(S,\mbf{L},\mbf{H})$ language \cite{GJ09} can be described in the annihilation-creation form
\begin{equation}\label{system:ani}\begin{aligned}
\dot{\breve{\mbf{a}}}&=\mathcal{A}\breve{\mbf{a}}+\mathcal{B}\breve{\mbf{b}}_{\rm in}, \\
\breve{\mbf{b}}_{\rm out}&=\mathcal{C}\breve{\mbf{a}}+\mathcal{D}\breve{\mbf{b}}_{\rm in},
\end{aligned}\end{equation}
where the system matrices 
\begin{equation*}\begin{aligned}
&\mathcal{C}=\Delta(C_-,C_+), ~~ \mathcal{D}=\Delta(S,0), \\
&\mathcal{B}=-\mathcal{C}^\flat\mathcal{D}, ~~ \mathcal{A}=-\imath J_n\Omega-\frac{1}{2}\mathcal{C}^\flat\mathcal{C}.
\end{aligned}\end{equation*}

Alternatively, the linear quantum system \eqref{system:ani} is equivalent to the following
real quadrature operator representation
\begin{equation}\label{system:qua}\begin{aligned}
\dot{\mbf{x}}&=\mathbb{A}\mbf{x}+\mathbb{B}\mbf{u}, \\
\mbf{y}&=\mathbb{C}\mbf{x}+\mathbb{D}\mbf{u},
\end{aligned}\end{equation}
where system state, input, and output are
\begin{equation*}\begin{aligned}
\mbf{x}=\left[\begin{array}{c}
\mbf{q} \\
\mbf{p}
\end{array}\right], ~~ \mbf{u}=\left[\begin{array}{c}
\mbf{q}_{\rm in} \\
\mbf{p}_{\rm in}
\end{array}\right], ~~ \mbf{y}=\left[\begin{array}{c}
\mbf{q}_{\rm out} \\
\mbf{p}_{\rm out}
\end{array}\right],
\end{aligned}\end{equation*}
respectively, and the corresponding constant matrices are
\begin{equation}\la{eq:mar24_ABCD}\begin{aligned}
&\mathbb{D}=V_m \mathcal{D} V_m^\dagger =\left[
                                          \begin{array}{cc}
                                            \mathrm{Re}(S) & -\mathrm{Im}(S) \\
                                            \mathrm{Im}(S) & \mathrm{Re}(S) \\
                                          \end{array}
                                        \right], \\
                                        &\mathbb{C}=V_m \mathcal{C} V_n^\dagger =\left[
                                                                             \begin{array}{cc}
                                                                               \mathrm{Re}(C_-+C_+) & -\mathrm{Im}(C_--C_+) \\
                                                                               \mathrm{Im}(C_-+C_+) & \mathrm{Re}(C_--C_+) \\
                                                                             \end{array}
                                                                           \right], \\
&\mathbb{B}=V_n \mathcal{B} V_m^\dagger = -\left[
                                            \begin{array}{cc}
                                              \mathrm{Re}(C_-^\dagger-C_+^\dagger) & -\mathrm{Im}(C_-^\dagger-C_+^\dagger) \\
                                              \mathrm{Im}(C_-^\dagger+C_+^\dagger) & \mathrm{Re}(C_-^\dagger+C_+^\dagger) \\
                                            \end{array}
                                          \right]\mathbb{D}, \\
                                          &\mathbb{A} = V_n \mathcal{A} V_n^\dagger = \mathbb{J}_n\mathbb{H}-\frac{1}{2}\mathbb{C}^\sharp \mathbb{C},
\end{aligned}\end{equation}
with
\begin{equation}\la{eq:mar24_JH}
\mathbb{J}_n\mathbb{H}=\mathbb{J}_n V_n \Omega V_n^\dagger=\left[
                                                             \begin{array}{cc}
                                                               \mathrm{Im}(\Omega_-+\Omega_+) & \mathrm{Re}(\Omega_--\Omega_+) \\
                                                               -\mathrm{Re}(\Omega_-+\Omega_+) & \mathrm{Im}(\Omega_--\Omega_+) \\
                                                             \end{array}
                                                           \right],
\end{equation}
and
\begin{equation}\label{eq:apr_C_sharp_C}
\mathbb{C}^\sharp\mathbb{C}=
\begin{bmatrix}
\begin{aligned}
&\mathrm{Re}(C_-^\dagger C_--C_+^\dagger C_++C_-^\dagger C_+-C_+^\top C_-^\#) \\
&\quad \quad \quad -\mathrm{Im}(C_-^\dagger C_-+C_+^\dagger C_++C_-^\dagger C_+-C_+^\top C_-^\#) \\[6pt]
&\mathrm{Im}(C_-^\dagger C_-+C_+^\dagger C_+-C_-^\dagger C_++C_+^\top C_-^\#) \\
&\quad \quad \quad \mathrm{Re}(C_-^\dagger C_--C_+^\dagger C_+-C_-^\dagger C_++C_+^\top C_-^\#)
\end{aligned}
\end{bmatrix}.
\end{equation}

In the following subsections, we respectively address the sufficient conditions under which bilateral and unilateral quantum BAE measurements can be implemented.


\subsection{Bilateral quantum BAE measurements}\label{sec:bilat}

In this subsection, we focus on the roles played by the system parameters $S$, $\mathcal{C}$, and $\Omega$ in achieving bilateral quantum BAE measurements, where the cross-coupling between conjugate input and output quadratures vanishes in both directions. By imposing structural conditions on these matrices, the transfer matrix becomes block-diagonal, thereby enabling the simultaneous evasion of back-action noise on both quadratures.

Firstly, we assume that $\Omega$ is purely imaginary and $\mathcal{C}$ is real. The scattering operator $S$ is real, e.g., $S=I$ for simplicity. The following result can be derived directly.

\bprop\label{prop:BAE}
If $\Omega$ is purely imaginary, both $S$ and $\mathcal{C}$ are real, then the transfer matrix is of the form
\beq\label{Mar15-1}
\begin{aligned}
&\mathbb{G}[s] =  \left[ 
\bey{cc}
\mathbb{G}_q[s] & 0 \\
0  & \mathbb{G}_p[s] 
\eey
\right],
\end{aligned}
\eeq
where
\begin{equation*}
\begin{aligned}
\mathbb{G}_q[s]=
S-\mathbb{C}_q [sI+\imath(\Omega_-+\Omega_+)+\frac1{2}\mathbb{C}_p^\top \mathbb{C}_q]^{-1} \mathbb{C}_p^\top, \\
\mathbb{G}_p[s]=
S-\mathbb{C}_p [sI+\imath(\Omega_--\Omega_+)+\frac1{2}\mathbb{C}_q^\top \mathbb{C}_p]^{-1} \mathbb{C}_q^\top,
\end{aligned}    
\end{equation*}
with
\beq\label{eq:jun29_C+-}
\mathbb{C}_q \triangleq C_-+C_+,  \ \ \ 
\mathbb{C}_p \triangleq C_--C_+.
\eeq
Thus, BAE measurements of $\mbf{q}_{\rm out}$ with respect to $\mbf{p}_{\rm in}$ and $\mbf{p}_{\rm out}$ with respect to $\mbf{q}_{\rm in}$ are realized.
\eprop

\begin{proof}
In this case, Eq. \eqref{eq:mar24_JH} reduces to
\begin{equation}\label{eq:mar24_JH5}
\mathbb{J}_n\mathbb{H}=-\imath \left[
                         \begin{array}{cc}
                           \Omega_-+\Omega_+ & 0 \\
                           0 & \Omega_--\Omega_+ \\
                         \end{array}
                       \right].
\end{equation}
Similarly, by Eq. \eqref{eq:mar24_ABCD},
\begin{equation}\la{eq:mar24_BC}\begin{aligned}
&\mathbb{C}=V_m \mathcal{C} V_n^\dagger =\left[
                                                                             \begin{array}{cc}
                                                                              C_-+C_+ & 0 \\
                                                                               0 & C_--C_+ \\
                                                                             \end{array}
                                                                           \right], \\
&\mathbb{B}=V_n \mathcal{B} V_m^\dagger =-\left[
                                            \begin{array}{cc}
                                              (C_--C_+)^\top & 0 \\
                                            0& (C_-+C_+)^\top) \\
                                            \end{array}
                                          \right],
\end{aligned}\end{equation}
and accordingly Eq. \eqref{eq:apr_C_sharp_C} reduces to
\begin{equation}\la{eq:mar24_A}\begin{aligned}
\mathbb{C}^\sharp\mathbb{C}
 =\left[
                                                    \begin{array}{cc}
                                                     (C_--C_+)^\top (C_-+C_+) & 0 \\
                                                      0&(C_-+C_+)^\top  (C_--C_+)  \\
                                                    \end{array}
                                                  \right].
\end{aligned}\end{equation}
According to Eqs. \eqref{eq:mar24_JH5} and \eqref{eq:mar24_A}, we have
\beq\la{eqLmar34_A5}\begin{aligned}
\mathbb{A} = &-\imath \left[
                         \begin{array}{cc}
                           \Omega_-+\Omega_+ & 0 \\
                           0 & \Omega_--\Omega_+ \\
                         \end{array}
                       \right] \\
                       &-\frac{1}{2}\left[
                                                    \begin{array}{cc}
                                                     (C_--C_+)^\top (C_-+C_+) & 0 \\
                                                      0&(C_-+C_+)^\top  (C_--C_+)  \\
                                                    \end{array}
                                                  \right].
\end{aligned}\eeq 
Consequently, the system matrices all exhibit a block-diagonal structure, indicating that the dynamics of the position and momentum quadratures are completely decoupled. The resulting transfer function matrix $\mathbb{G}[s]$ therefore takes the block-diagonal form \eqref{Mar15-1}, which establishes the bilateral BAE measurements claimed in the proposition.
\end{proof}

\begin{remark}
Noticing that the system can be non-passive as $C_+$ and $\Omega_+$ can be nonzero in this case. The non-degenerate parametric amplifier (NDPA) (after rotation) in \cite{SY21}, the DPA, and the model studied in \cite{BQND24} (ignoring the perturbation term) belong to this special case.
\end{remark}

On the other hand, if $\mathcal{C}$ is purely imaginary in Proposition \ref{prop:BAE}, the transfer matrix in the real quadrature form is also block diagonal, which indicates that in this case the system position $\mbf{q}$ and momentum $\mbf{p}$ are separable. The following theorem summarizes the main results of the preceding two cases.

\begin{theorem}\label{lem:BAE1}
When $S$ is real, $\Omega$ is purely imaginary and $\mathcal{C}$ is real or purely imaginary, quantum BAE measurements of $\mbf{q}_{\rm out}$ with respect to $\mbf{p}_{\rm in}$ and $\mbf{p}_{\rm out}$ with respect to $\mbf{q}_{\rm in}$ are realized.
If further 
\begin{itemize}
\item $C_-=C_+$ ($\mbf{q}$ is coupled with the environment), then $\mbf{q}$ is a QND variable if the subsystem $(\imath(\Omega_-+\Omega_+),0,C_-)$ is observable;

\item $C_-=-C_+$ ($\mbf{p}$ is coupled with the environment), then $\mbf{p}$ is a QND variable if the subsystem $(\imath(\Omega_--\Omega_+),0,C_-)$ is observable.
\end{itemize}
\end{theorem}

\begin{proof}
The case where $\mathcal{C}$ is real has already been established in Proposition \ref{prop:BAE}. We now turn to the case where $\mathcal{C}$ is purely imaginary. In this case, the system matrices $\mathbb{C}$ and $\mathbb{B}$ can be described by the following block off-diagonal forms
\begin{equation*}\begin{aligned}
&\mathbb{C}=\left[\begin{array}{cc}
0 & -\mathrm{Im}(C_--C_+) \\
\mathrm{Im}(C_-+C_+) & 0
\end{array}\right], \\
&\mathbb{B}=\left[\begin{array}{cc}
0 & \mathrm{Im}(C_-^\dagger-C_+^\dagger) \\
-\mathrm{Im}(C_-^\dagger+C_+^\dagger) & 0
\end{array}\right].
\end{aligned}\end{equation*}
Then
\begin{equation*}\begin{aligned}
\mathbb{C}^\sharp\mathbb{C}=
\left[\begin{array}{cc}
(C_-^\dagger - C_+^\dagger)(C_- + C_+) & 0 \\
0 & (C_-^\dagger + C_+^\dagger)(C_- - C_+)
\end{array}\right],
\end{aligned}\end{equation*} 
which means that the system matrix $\mathbb{A}$ is block diagonal. Consequently, the transfer matrix is also block diagonal. Specifically, the dynamical evolution of this system is expressed by
\begin{equation}\label{Jan9-1}\begin{aligned}
\dot{\mbf{q}} =& \left(-\imath(\Omega_-+\Omega_+)-
\frac{1}{2}(C_- - C_+)^\dagger(C_-+C_+)\right)\mbf{q} \\
&+ \mathrm{Im}(C_--C_+)^\dagger\mbf{p}_{\rm in}, \\
\dot{\mbf{p}} =& \left(-\imath(\Omega_--\Omega_+)-
\frac{1}{2}(C_- + C_+)^\dagger(C_--C_+)\right)
\mbf{p} \\
&-\mathrm{Im}(C_-+C_+)^\dagger\mbf{q}_{\rm in}, \\
\mbf{q}_{\mathrm{out}}=&\; -\mathrm{Im}(C_--C_+)\mbf{p} +\mbf{q}_{\rm in}, \\
\mbf{p}_{\mathrm{out}} =&\; \mathrm{Im}(C_-+C_+)\mbf{q} + \mbf{p}_{\rm in},
\end{aligned}\end{equation}
where quantum BAE measurements of $\mbf{q}_{\rm out}$ with respect to $\mbf{p}_{\rm in}$ and $\mbf{p}_{\rm out}$ with respect to $\mbf{q}_{\rm in}$ are realized \cite[Theorem 4.1]{ZPL20}. When $C_-=C_+$, $\mbf{q}$ is a QND variable if the subsystem $(\imath(\Omega_-+\Omega_+),0,C_-)$ is observable; when $C_-=-C_+$, $\mbf{p}$ is a QND variable if the subsystem $(\imath(\Omega_--\Omega_+),0,C_-)$ is observable \cite[Remark 4.9]{ZGPG18}.
\end{proof}

\begin{remark}
In essence, Theorem \ref{lem:BAE1} establishes that a purely oscillatory internal dynamics (a purely imaginary Hamiltonian matrix) combined with a matched, real or imaginary system-probe coupling is sufficient to decouple the system's response to back-action noise, enabling bilateral evasion.    
\end{remark}

Secondly, we consider the scenario in which the scattering matrix $S$ is purely imaginary. The following theorem states that, under analogous structural conditions on $\mathcal{C}$ and $\Omega$, bilateral BAE measurements can similarly be realized.

\begin{theorem}\label{lem:BAE2}
When both $S$ and $\Omega$ are purely imaginary, $\mathcal{C}$ is real or purely imaginary, quantum BAE measurements of  $\mbf{q}_{\rm out}$ with respect to $\mbf{q}_{\rm in}$ and $\mbf{p}_{\rm out}$ with respect to $\mbf{p}_{\rm in}$  are realized.
If further 
\begin{itemize}
\item $C_-=C_+$ ($\mbf{q}$ is coupled with the environment), then $\mbf{q}$ is a QND variable if the subsystem $(\imath(\Omega_-+\Omega_+),0,C_-)$ is observable;
\item $C_-=-C_+$ ($\mbf{p}$ is coupled with the environment), then $\mbf{p}$ is a QND variable if the subsystem $(\imath(\Omega_--\Omega_+),0,C_-)$ is observable.
\end{itemize}
\end{theorem}

\begin{proof}
Since $S$ is purely imaginary, the system matrix $\mathbb{D}$ is block off-diagonal. We first assume that $\mathcal{C}$ is real.  
In this case,
\begin{equation*}
\mathbb{C} = \left[\begin{array}{cc}
C_-+C_+ & 0 \\
0 & C_--C_+
\end{array}\right], \ \ \
\mathbb{B} = \left[\begin{array}{cc}
0 & -(C_--C_+)^\top \\
(C_-+C_+)^\top & 0
\end{array}\right].    
\end{equation*}
Together with the block-diagonal form \eqref{eqLmar34_A5} of the system matrix $\mathbb{A}$ implied by the purely imaginary $\Omega$, the system matrices are all block off-diagonal or block diagonal, 
so that the resulting transfer matrix $\mathbb{G}[s]$ 
is also block off-diagonal. Consequently, the cross-coupling terms 
$\mathbb{G}_{\mathbf{q}_{\rm in}\to\mathbf{q}_{\rm out}}[s]$ and 
$\mathbb{G}_{\mathbf{p}_{\rm in}\to\mathbf{p}_{\rm out}}[s]$ vanish, 
which establishes the bilateral BAE measurements for the quadrature pairs 
$(\mathbf{q}_{\rm in},\mathbf{q}_{\rm out})$ and $(\mathbf{p}_{\rm in},\mathbf{p}_{\rm out})$. An analogous argument applies when $\mathcal{C}$ is purely imaginary, leading to the same block off-diagonal structure of the transfer matrix $\mathbb{G}[s]$ and hence to the same BAE measurements.

The QND properties for observables $\mathbf{q}$ and $\mathbf{p}$ follow by combining the subsystem conditions with the block-diagonal dynamics, exactly as in the discussion given by Theorem \ref{lem:BAE1}.     
\end{proof}


\subsection{Unilateral quantum BAE measurement}\label{sec:unilat}

In this subsection, we relax the requirement that $\Omega$ be purely imaginary and instead assume that the real parts of $\Omega_-$ and $\Omega_+$ are either equal or opposite in sign. Under these weaker conditions the transfer matrix becomes block‑triangular, which yields \emph{unilateral} BAE measurements.  
We also treat the important special case in which the coupling matrix $\mathcal{C}$ is purely imaginary and satisfies $C_-=\pm C_+$, i.e., $\mathbf{q}$‑ or $\mathbf{p}$‑coupling, for which the BAE condition holds irrespective of $\Omega$.

\subsubsection{\texorpdfstring{Dependence on the real parts of $\Omega$}{Dependence on the real parts of Omega}}\label{sec:unilat-real}

When the real parts satisfy $\mathrm{Re}(\Omega_-)=\pm\mathrm{Re}(\Omega_+)$, the matrix $\mathbb{J}_n\mathbb{H}$ becomes block‑triangular, and therefore so does the system matrix $\mathbb{A}$.  
As a consequence, the transfer matrix $\mathbb{G}[s]$ inherits one of four possible block‑triangular structures, depending on whether $S$ and $\mathcal{C}$ are real or purely imaginary.  
Tables \ref{tab:same-real} and \ref{tab:opp-real} summarize the resulting unilateral BAE relations for the two sign choices.

\begin{table}[htb]
\centering
\caption{Unilateral BAE conditions when $\mathrm{Re}(\Omega_-) = \mathrm{Re}(\Omega_+)$.}
\label{tab:same-real}
\begin{tabular}{c|c|c}
\toprule
$S$ & $\mathcal{C}$ & Vanishing transfer block (BAE) \\
\midrule
real           & real           & $\mathbb{G}_{\mathbf{p}_{\rm in}\to\mathbf{q}_{\rm out}}[s]=0$ \\[2pt]
real           & purely imag.   & $\mathbb{G}_{\mathbf{q}_{\rm in}\to\mathbf{p}_{\rm out}}[s]=0$ \\[2pt]
purely imag.   & real           & $\mathbb{G}_{\mathbf{q}_{\rm in}\to\mathbf{q}_{\rm out}}[s]=0$ \\[2pt]
purely imag.   & purely imag.   & $\mathbb{G}_{\mathbf{p}_{\rm in}\to\mathbf{p}_{\rm out}}[s]=0$ \\
\bottomrule
\end{tabular}
\end{table}

\begin{table}[htb]
\centering
\caption{Unilateral BAE conditions when $\mathrm{Re}(\Omega_-) = -\mathrm{Re}(\Omega_+)$.}
\label{tab:opp-real}
\begin{tabular}{c|c|c}
\toprule
$S$ & $\mathcal{C}$ & Vanishing transfer block (BAE) \\
\midrule
real           & real           & $\mathbb{G}_{\mathbf{q}_{\rm in}\to\mathbf{p}_{\rm out}}[s]=0$ \\[2pt]
real           & purely imag.   & $\mathbb{G}_{\mathbf{p}_{\rm in}\to\mathbf{q}_{\rm out}}[s]=0$ \\[2pt]
purely imag.   & real           & $\mathbb{G}_{\mathbf{p}_{\rm in}\to\mathbf{p}_{\rm out}}[s]=0$ \\[2pt]
purely imag.   & purely imag.   & $\mathbb{G}_{\mathbf{q}_{\rm in}\to\mathbf{q}_{\rm out}}[s]=0$ \\
\bottomrule
\end{tabular}
\end{table}

The derivations of these blocking patterns follow exactly the same steps as in Section \ref{sec:bilat} and are omitted for brevity.

\subsubsection{\texorpdfstring{The special case of $\mathbf{q}$‑ or $\mathbf{p}$‑coupling}{The special case of q‑ or p‑coupling}}\label{sec:unilat-qp}

We now drop all constraints on $\Omega$ and consider instead the coupling matrix $\mathcal{C}$ that is purely imaginary and satisfies $C_-=\pm C_+$ (i.e., only the $\mathbf{q}$ or the $\mathbf{p}$ quadratures are coupled to the probe). The following corollary states the resulting BAE relations, its proof is a straightforward consequence of the block‑triangular structure of the transfer matrix and is therefore omitted.

\begin{corollary}\label{cor:BAE}
Let $S$ be real and $\mathcal{C}$ be purely imaginary.
\begin{itemize}
  \item If $C_- = C_+$ ($\mathbf{q}$‑coupling), quantum BAE measurement of $\mathbf{q}_{\rm out}$ with respect to $\mathbf{p}_{\rm in}$ is realized.
  \item If $C_- = -C_+$ ($\mathbf{p}$‑coupling), quantum BAE measurement of $\mathbf{p}_{\rm out}$ with respect to $\mathbf{q}_{\rm in}$ is realized.
\end{itemize}
\end{corollary}

The following example, taken from \cite[Fig. 3(c)]{NY14} and \cite[Eq. (35)]{ZPL20}, illustrates Corollary \ref{cor:BAE}.

\begin{example}[Michelson interferometer]
Consider a Michelson interferometer, one of the simplest devices for gravitational‑wave detection, with system Hamiltonian
\begin{equation*}\begin{aligned}
\Omega_-=\frac{1}{2}\left[\begin{array}{cc}
m \omega_m^2+\frac{1}{m} & 0 \\
0 & m \omega_m^2+\frac{1}{m}
\end{array}\right], \\ 
\Omega_+=\frac{1}{2}\left[\begin{array}{cc}
m \omega_m^2-\frac{1}{m} & 0 \\
0 & m \omega_m^2-\frac{1}{m}
\end{array}\right],
\end{aligned}\end{equation*}
where $m$ and $\omega_m$ denote the mass of the mechanical
oscillators and the resonant frequency, respectively. The coupling operator matrix $\mathcal{C}$ is purely imaginary with 
\begin{equation*}
C_-=C_+=\frac{\imath\sqrt{\lambda}}{2}\left[\begin{array}{cc}
1 & 1 \\
1 & -1
\end{array}\right],    
\end{equation*}
where $\lambda$ is the coupling strength between the input field and the mechanical oscillators. The scattering matrix $S$ is chosen to be identity, i.e., $S= I$. As a result, the transfer function of the Michelson interferometer $\mathbb{G}[s]$ can be calculated as a block lower triangular matrix, which confirms that quantum BAE measurement of $\mbf{q}_{\rm out}$ with respect to $\mbf{p}_{\rm in}$ is realized.   
\end{example}


\section{QND interaction}\label{QND INT}

A QND interaction enables the repeated measurement of a quantum observable without perturbing its value or its eigenstate distribution. This is achieved by designing the system–probe coupling such that the measurement operator commutes with the system Hamiltonian. For example, consider a set of mutually commuting self-adjoint coupling operators $\mbf{L} = [\mbf{L}_1, \ldots , \mbf{L}_m]^\top$, which can be used as measurement operators for a continuous-time projection measurement process. $[\mbf{L},\mbf{H}]=0$ is assumed in the so-called QND interaction \cite{BK95,QPG13}, \cite[Sec. 4.2.4]{NY17}, \cite{GZP19,CSR+20}, and thus the time evolution of the measurement operator $\mbf{L}$ can be calculated by Eq. \eqref{dX} ($S=I$)
\begin{equation}\label{eq:L_sept16}
d\mbf{L}_j= (-\imath [\mbf{L}_j, \mbf{H}]+\frac{1}{2}([\mbf{L}^\dagger,\mbf{L}_j]\mbf{L})dt+
[\mbf{L}^\dagger,\mbf{L}_j]d\mbf{B}_{\rm in}=0,
\end{equation}
which means that the observable $\mbf{L}_j$ remains as its initial value and is independent of the measurement process.  Moreover the transfer matrix is the identity matrix, therefore quantum BAE measurements are realized.

\bmrk
According to Eq. \eqref{eq:L_sept16}, $\mbf{L}$ is set of mutually commuting observables, then they are QND variables. Moreover, as the transfer function is an identity matrix, therefore, quantum BAE measurements are realized.  Thus, if $\mbf{L}$ is a set of mutually commuting observables and $[\mbf{L},\mbf{H}]=0$, then QND and BAE are realized simultaneously. 
\emrk

Belavkin quantum filtering equation, which is a classical stochastic differential equation, is given by
\begin{equation}\label{dec19_1}\begin{aligned}
 d\pi_t(\mbf{L}_j)
=&  \pi_t(\mathcal{L}(\mbf{L}_j))dt + \Big[\pi_t(\mbf{L}_j\mbf{L}^\top+\mbf{L}^\dagger \mbf{L}_j) \\
&-\pi_t(\mbf{L}^\top+\mbf{L}^\dagger)\pi_t(\mbf{L}_j)\Big]d\nu(t),
\end{aligned}\end{equation}
where $\nu(t)$ is Wiener process \cite[Sec. 6.1]{BvHJ07} and the superoperator 
\begin{equation*}
\mathcal{L}(\mbf{L}_j)=-\imath[\mbf{L}_j,\mbf{H}]+\frac{1}{2}\mbf{L}^\dagger[\mbf{L}_j,\mbf{L}]+\frac{1}{2}[\mbf{L}^\dagger,\mbf{L}_j]\mbf{L} =0.    
\end{equation*}
Here, we assume the initial conditional state has the finite first and second moments, $\pi_0(\mbf{L}_j) = \mathrm{Tr}(\rho_{\rm S}(0)\mbf{L}_j)$, where $\rho_{\rm S}(0)$ is the initial system density state. Thus, we have $d\pi_t(\mbf{L}_j)
=  [\pi_t(\mbf{L}_j\mbf{L}^\top+\mbf{L}^\dag \mbf{L}_j) -\pi_t(\mbf{L}_j)\pi_t(\mbf{L}^\top+\mbf{L}^\dag)]d\nu(t)$. Taking expectation with respect to the Wiener process $\nu(t)$, yields that  $\mathbb{E}[ d\pi_t(\mbf{L}_j)]=0$. As a result,
\begin{equation}\label{eq:sept18_QND}
\mathbb{E}[\pi_t(\mbf{L}_j)] = \mathbb{E}[\pi_0(\mbf{L}_j)]  = {\rm Tr}(\rho_S(0)\mbf{L}_j).    
\end{equation}
Moreover, $\mathbb{E}[\pi_t(\mbf{L}_j^2)] = \mathbb{E}[\pi_0(\mbf{L}_j^2)]  = {\rm Tr}(\rho_S(0)\mbf{L}_j^2)$.

\bmrk
According to Eq. \eqref{eq:sept18_QND}, under the QND interaction, measurements have no back action on the estimation of $\mbf{L}_j$. 
\emrk

In the following SISO case, we go to the Schr{\"o}dinger picture. Let $\mbf{L}$ be a (possibly unbounded) self-adjoint operator,
and let $P_{\mbf{L}}(\cdot)$ denote its spectral measure. Thus, by von Neumann's spectral theorem, we have
\begin{equation*}
  \mbf{L} = \int_{\mathbf{R}} \lambda \, dP_{\mbf{L}}(\lambda).
\end{equation*}
For any Borel set $\mathcal{E}\subset\mathbf{R}$, define the spectral probability
\begin{equation*}
  \mu_t(\mathcal{E}) \triangleq \mathrm{Tr}\!\left[\rho_c(t)\,P_{\mbf{L}}(\mathcal{E})\right].
\end{equation*}
Clearly, $\mu_t(\mathcal{E})\in[0,1]$, $\mu_t(\mathbf{R})=1$, and $\mathcal{E}\mapsto \mu_t(\mathcal{E})$ is a
probability measure for each fixed $t$.
Under the assumption $[\mbf{L},\mbf{H}]=0$,  $P_{\mbf{L}}(\mathcal{E})$ is preserved by the Heisenberg dynamics. By applying QSME to the bounded projector $P_{\mbf{L}}(\mathcal{E})$ yields a following stochastic differential equation:
\begin{equation*}
  d\mu_t(\mathcal{E}) = {\rm Tr}\big[\mathcal{G}_{\mathbf{L}}\big(\rho_c(t)\big)P_{\mbf{L}}(\mathcal{E}) \big] d\nu(t),
\end{equation*}
where $\mathcal{G}_{\mathbf{L}}(\rho)$ is the diffusion term of QSME \eqref{Jan12-2}.
In particular,  for each fixed $\mathcal{E}$, $\mu_t(\mathcal{E})$ is a bounded martingale, and there exist random variables $\mu_{\infty}(\mathcal{E})\in[0,1]$ such that
\begin{equation*}
    \lim_{t\rightarrow \infty}\mu_t(\mathcal{E})=\mu_{\infty}(\mathcal{E})
\end{equation*}
almost surely and in $L^1$ norm. Moreover
\begin{equation*}
 \mathbb{E}\left[\mu_t(\mathcal{E})\right] =\mu_0(\mathcal{E}),\quad \forall\, t\ge 0.
\end{equation*}
Moreover, $\mu_\infty(\mathcal E)$ is $\{0,1\}$-valued provided that different spectral
values of ${\mbf{L}}$ induce distinguishable long-time measurement statistics.
If ${\mbf{L}}=\sum_{j}\lambda_j P_j$ has a discrete spectrum with pairwise distinct eigenvalues
$\lambda_j\neq\lambda_k$ for $j\neq k$,
we have 
\begin{equation*}
    \mathbb{P}\big(\mu_{\infty}(\lambda_j)=1 \big)=\mathrm{Tr}(\rho_c(0) P_j), \quad \mathbb{P}\big(\mu_{\infty}(\lambda_j)\mu_{\infty}(\lambda_k)=0 \big)=1.
\end{equation*}
It means the conditional state $\rho_c(t)$ converges to one of eigenspaces of ${\mbf{L}}$ almost surely, which is referred to \emph{quantum state reduction} \cite{TC2014}, \cite[Thm. 5.1]{LAM19}, \cite[Lem. 1]{CSR+20}.

\subsection{The SISO case}\label{siso}

In this subsection, we analyze the SISO case, where the BAE condition reduces to two simple commutator relations, namely $[\mathbf{L}+\mathbf{L}^\ast,\mathbf{H}]=0$ or $[\mathbf{L}-\mathbf{L}^\ast,\mathbf{H}]=0$. We formalize this result as a theorem, whose direct proof  is given below.

\begin{theorem}\label{thm:BAE_SISO}
Quantum BAE measurement in the SISO case can be realized if 
\[
[\mathbf{L}+\mathbf{L}^\ast,\mathbf{H}]=0 \quad\text{or}\quad [\mathbf{L}-\mathbf{L}^\ast,\mathbf{H}]=0.
\]
\end{theorem}

\begin{proof}
Since the number of input channels is $m=1$, it is straightforward to verify that $C_-C_+^\top = C_+C_-^\top$. We treat the two conditions separately.

\noindent\textit{Case 1:} $[\mathbf{L}+\mathbf{L}^\ast,\mathbf{H}]=0$.
From the definition $\mathbf{L}=C_-\mathbf{a}+C_+\mathbf{a}^\#$ and the elementary commutators
\begin{align*}
[C_-^\#\mathbf{a}^\#,C_-\mathbf{a}] &= -\sum_j |C_{-,j}|^2, \\
[C_+^\#\mathbf{a},C_+\mathbf{a}^\#] &= \sum_k |C_{+,k}|^2,
\end{align*}
a short calculation yields the stochastic differential equation
\begin{equation}\label{eq:sep12_1}
d(\mathbf{L}+\mathbf{L}^\ast) = -\frac{g}{2}(\mathbf{L}+\mathbf{L}^\ast)dt - \sqrt{2}\,g\,d\mathbf{Q}_{\mathrm{in}},
\end{equation}
with
\begin{equation}\label{sepg}
g \triangleq \sum_{j=1}^n (|C_{-,j}|^2-|C_{+,j}|^2),\qquad
\mathbf{Q}_{\mathrm{in}} = \frac{1}{\sqrt{2}}\bigl(\mathbf{B}_{\mathrm{in}}+\mathbf{B}_{\mathrm{in}}^\ast\bigr).   
\end{equation}
Eq. \eqref{eq:sep12_1} can be integrated explicitly,
\[
(\mathbf{L}+\mathbf{L}^\ast)(t) = e^{-gt/2}(\mathbf{L}+\mathbf{L}^\ast)(0) - \sqrt{2}\,g\int_0^t e^{-g(t-\tau)/2}d\mathbf{Q}_{\mathrm{in}}(\tau),
\]
and the input–output relation 
\begin{equation}\label{sep11-1}
\mathbf{q}_{\rm out}(t) = \frac{1}{\sqrt{2}}(\mathbf{L}+\mathbf{L}^\ast)(t) + \mathbf{q}_{\mathrm{in}}(t)    
\end{equation}
leads, after Laplace transformation, to
\begin{equation}\label{eq:sept16_q_out}
\mathbf{q}_{\rm out}[s] = \left[1-g\left(s+\frac{g}{2}\right)^{-1}\right]\mathbf{q}_{\mathrm{in}}[s]
= \frac{s-g/2}{s+g/2}\,\mathbf{q}_{\mathrm{in}}[s].    
\end{equation}
Hence the output quadrature $\mathbf{q}_{\rm out}$ is independent of the input quadrature $\mathbf{p}_{\mathrm{in}}$, which is precisely the BAE condition $\mathbb{G}_{\mathbf{p}_{\mathrm{in}}\to\mathbf{q}_{\mathrm{out}}}[s]=0$.

\noindent\textit{Case 2:} $[\mathbf{L}-\mathbf{L}^\ast,\mathbf{H}]=0$. A similar computation gives
\begin{equation}\label{Oct16-1}
d(\mathbf{L}-\mathbf{L}^\ast) = -\frac{g}{2}(\mathbf{L}-\mathbf{L}^\ast)dt - \imath\sqrt{2}\,g\,d\mathbf{P}_{\mathrm{in}},
\end{equation}
where $\mathbf{P}_{\mathrm{in}} = \frac{-\imath}{\sqrt{2}}(\mathbf{B}_{\mathrm{in}}-\mathbf{B}_{\mathrm{in}}^\ast)$.
Solving Eq. \eqref{Oct16-1}, yields that 
\begin{equation*}
(\mbf{L-L^\ast})(t) = e^{-gt/2}(\mbf{L-L^\ast})(0)-\imath\sqrt{2}g \int_0^t e^{-g(t-\tau)/2} d\mbf{P}_{\mathrm{in}}(\tau).    
\end{equation*}
We have
\beq\label{sep11-2}
\mbf{p}_{\rm out}(t) =\frac{-\imath}{\sqrt{2}}(\mbf{L-L^\ast})(t) + \mbf{p}_{\mathrm{in}}(t), 
\eeq
and thus
\begin{equation*}
\mbf{p}_{\rm out}[s] = [1-g (s+\frac{g}{2})^{-1}] \mbf{p}_{\mathrm{in}}[s] = \frac{s-g/2}{s+g/2}\mbf{p}_{\mathrm{in}}[s],     
\end{equation*}
where the output quadrature $\mbf{p}_{\rm out}$ is insensitive to the input quadrature $\mathbf{q}_{\mathrm{in}}$.

Both cases together establish the unilateral BAE measurements described in the theorem.
\end{proof}

\begin{remark}
If $\sum_j |C_{-,j}|^2=\sum_k |C_{+,k}|^2$, then $g=0$ in Eq. \eqref{sepg}, which yields that $d(\mbf{L}+\mbf{L}^\ast)=0$ and the coupling operator $(\mbf{L}+\mbf{L}^\ast)(t)=(\mbf{L}+\mbf{L}^\ast)(0)=\mbf{L}+\mbf{L}^\ast$ can be detected via the input-output relation \eqref{sep11-1}. Thus, this kind of interaction does not change the physical property of the coupling operator $\mbf{L}+\mbf{L}^\ast$,
which is called a QND interaction in \cite[Sec. 4.2.4]{NY17}. On the other hand, by Eq. \eqref{sep11-1},  $\mbf{L}+\mbf{L}^\ast$ is observable, and thus it is a QND variable. Similar results are also applied to the combined Hermitian coupling operator $\imath(\mbf{L}-\mbf{L}^\ast)(t) \equiv \imath(\mbf{L}-\mbf{L}^\ast)$, in which the QND interaction can be observed by Eq. \eqref{Oct16-1} and the input-output relation \eqref{sep11-2}. Indeed, $\mbf{L}=\mbf{L}^\ast$ assumed in \cite[Sec. 4.2.4]{NY17} is a more special case, which realizes the condition $g=0$ in Eq. \eqref{sepg}, and bilateral BAE measurements are realized. In that special case, $[\mbf{L}+\mbf{L}^\ast,\mbf{H}]=0$ reduces to $[\mbf{L},\mbf{H}]=0$. 
\end{remark}

\bmrk
In the SISO case, let $\mbf{L}$ be an observable and $[\mbf{L},\mbf{H}]=0$, then $\mbf{L}$ is a QND variable and BAE measurement is also realized. If $\mbf{L}$ is not an observable, by Eq. \eqref{eq:sept16_q_out} BAE measurements can still be implemented. However conditions like  $\sum_j |C_{-,j}|^2=\sum_k |C_{+,k}|^2$ should be added to have QND variables.
\emrk

In the following subsections, we focus on the MIMO quantum systems and discuss the properties of $[\mbf{L},\mbf{H}]=0$ in the annihilation-creation, and Kalman canonical forms, respectively.


\subsection{The MIMO case}

In the annihilation-creation operator form, $\mbf{L}=\mbf{L}^\#$ is equivalent to  $C_-=C_+^\#$. In particular, if both $C_-$ and $C_+$  are real, then $\mbf{L}=\mbf{L}^\#$ is equivalent to $C_-=C_+$. 

\bmrk
It is obvious that $[\mbf{L},\mbf
{H}]=0$ implies that $[\mbf{L}^\#,\mbf{H}]=0$, and thus $[\mbf{L}\pm\mbf{L}^\#,\mbf{H}]=0$. However, either $[\mbf{L}+\mbf{L}^\#,\mbf{H}]=0$ or  $[\mbf{L}-\mbf{L}^\#,\mbf{H}]=0$ alone does not necessarily lead to $[\mbf{L},\mbf
{H}]=0$.
\emrk

The following lemma will be helpful to derive the equations given in Proposition \ref{Jan29-1}.
\begin{lemma}\cite{ZD22}\label{Jan29-2}
Let $\mbf{X}$, $\mbf{Y}$, and $\mbf{Z}$ be vectors of operators with dimension $l$, $m$, and $n$, respectively. Let $M\in\mathbf{C}^{m\times n}$, and the commutators $[\mbf{a},\mbf{b}]\in\mathbf{C}$ where $\mbf{a}$ and $\mbf{b}$ are arbitrary elements of the vectors $\mbf{X}$, $\mbf{Y}$ and $\mbf{Z}$. Then
\begin{equation*}
\left[\mbf{X},\mbf{Y}^\top M \mbf{Z}\right]=\left[\mbf{X},\mbf{Y}^\top\right]M\mbf{Z}+\left[\mbf{X},\mbf{Z}^\top\right]M^\top\mbf{Y}.    
\end{equation*} 
\end{lemma}

\bprop\label{Jan29-1}
$[\mbf{L},\mbf{H}]=0$ if and only if
\begin{equation}\label{Jan29-3}\begin{aligned}
C_-\Omega_-=C_+\Omega_+^\dagger, \ \ \ \
C_-\Omega_+=C_+\Omega_-^\top,
\end{aligned}\end{equation}
or equivalently, 
\begin{equation}\label{COmega}
\mathcal{C}\Omega=2\Delta(C_-,0)\Omega.
\end{equation}
\eprop

\begin{proof}
Inserting $\mbf{L}=C_-\mbf{a}+C_+\mbf{a}^\#$ and $\mbf{H}=\frac{1}{2}\breve{\mbf{a}}^\dagger\Omega\breve{\mbf{a}}$ into the commutator $\left[\mbf{L},\mbf{H}\right]$, by Lemma \ref{Jan29-2} we have
\begin{equation}\la{eq:jun29_LH}
\begin{aligned}
[\mbf{L},\mbf{H}]&=(C_-\Omega_--C_+\Omega_+^\dagger)\mbf{a}+(C_-\Omega_+-C_+\Omega_-^\top)\mbf{a}^\#, \\
[\mbf{L}^\#,\mbf{H}]&=(C_+^\#\Omega_--C_-^\#\Omega_+^\dagger)\mbf{a}+(C_+^\#\Omega_+-C_-^\#\Omega_-^\top)\mbf{a}^\#.
\end{aligned}\end{equation}
Thus, $\left[\mbf{L},\mbf{H}\right]=0$ if and only if Eq. \eqref{Jan29-3} holds, and is equivalent to Eq. \eqref{COmega} due to the Hermitian property of $\Omega$.  
\end{proof}

The condition $[\mathbf{L},\mathbf{H}] = 0$ established in Proposition \ref{Jan29-1} leads to several parameter degeneracies that are worth analyzing separately. These can be summarized compactly as the following corollary.  
We provide a proof for the representative cases $C_+=0$ and $\Omega_+=0$, the remaining two cases follow from analogous arguments and are therefore omitted for brevity.

\begin{corollary}\label{cor:blockdiag}
Assume that $[\mathbf{L},\mathbf{H}] = 0$ and that one of the following
conditions holds\/:
\begin{enumerate}
\item $C_+ = 0$,
\item $C_- = 0$,
\item $\Omega_+ = 0$ and $C_-C_+^\top$ is symmetric,
\item $\Omega_- = 0$ and $C_-C_+^\top$ is symmetric.
\end{enumerate}
Then the transfer matrix $\mathbb{G}[s]$ is block diagonal, and consequently
quantum BAE measurements are realized.
\end{corollary}

\begin{proof}
\noindent\textit{Case 1: $C_+ = 0$.} By Proposition \ref{Jan29-1}, the condition $[\mathbf{L},\mathbf{H}] = 0$ reduces to 
$C_-\Omega_- = 0$ and $C_-\Omega_+ = 0$, which implies 
$\mathcal{C}\Omega =0$.  
Using the expansion 
$\Sigma[s] = \frac{1}{2s}\mathcal{C}\sum_{n=0}^{\infty}(\frac{-\imath}{s})^{n}
(J_n\Omega)^{n}\mathcal{C}^{\flat}$ as in \cite{GZ15}, we obtain 
$\mathcal{C}J_n\Omega = 0$ and therefore 
$\Sigma[s] = \frac{1}{2s}\mathcal{C}\mathcal{C}^{\flat}$.
The transfer function becomes
\begin{equation*}\begin{aligned}
\mathbb{G}[s]=\mathrm{diag}\bigg\{&
(sI-\frac{1}{2}C_-C_-^\dagger)(sI+\frac{1}{2}C_-C_-^\dagger)^{-1}, \\
&(sI-\frac{1}{2}C_-^\#C_-^\top)(sI+\frac{1}{2}C_-^\#C_-^\top)^{-1} \bigg\},
\end{aligned}\end{equation*}
which is block diagonal, so BAE measurements are realized.

\noindent\textit{Case 3: $\Omega_+ = 0$ and $C_-C_+^\top$ is symmetric.}
Here $[\mathbf{L},\mathbf{H}] = 0$ gives 
$C_-\Omega_- = 0$ and $C_+\Omega_-^\top = 0$, 
which again yields $\mathcal{C}\Omega =0$.  
A similar calculation leads to 
$\mathcal{C}J_n\Omega = 0$ and $\Sigma[s] = \frac{1}{2s}\mathcal{C}\mathcal{C}^{\flat}$.
With the additional symmetry condition $C_-C_+^\top = C_+C_-^\top$, 
the transfer matrix simplifies to the block diagonal form
\begin{equation*}\begin{aligned}
&\mathbb{G}[s] \\
=&\mathrm{diag}\bigg\{(sI-\frac{1}{2}(C_-C_-^\dagger-C_+C_+^\dagger))(sI+\frac{1}{2}(C_-C_-^\dagger-C_+C_+^\dagger))^{-1}, \\
&(sI+\frac{1}{2}(C_+^\#C_+^\top-C_-^\#C_-^\top))(sI-\frac{1}{2}(C_+^\#C_+^\top-C_-^\#C_-^\top))^{-1}\bigg\},
\end{aligned}\end{equation*}
and BAE measurements follow.
\end{proof}

\bmrk
The condition $C_-C_+^\top$ being symmetric holds if $C_-=\pm C_+$ ($\mbf{q}$ or $\mbf{p}$ coupling)  or the quantum system is with the SISO case (see Section \ref{siso} for details).
\emrk

\bmrk
In the above four special cases with $[\mbf{L},\mbf{H}]=0$, we all have $\mathcal{C}\Omega=0$, which means that the transfer function $\mathbb{G}[s]$ does not depend on $\Omega$. 
\emrk

\subsection{The quantum Kalman canonical form}

As introduced in Section \ref{preliminaries}, the dynamics of a linear quantum system can be represented in the real quadrature form \eqref{system:qua}. Utilizing the concepts of controllability and observability \cite[Def. 1]{GZ15}, \cite[Def. 3.1]{ZD22}, \cite[Def. 2.1]{DZLP26}, a special realization known as the \emph{quantum Kalman canonical form} was derived in~\cite[Thm.~4.4]{ZGPG18}. A distinctive feature of this decomposition is that, due to the physical realizability constraints, the classical ``$c\bar{o}$'' (controllable but unobservable) and ``$\bar{c}o$'' (uncontrollable but observable) subsystems are always paired. They do not exist independently; instead, their modes are conjugately coupled to form an ``$h$'' subsystem. Consequently, the whole system is decomposed into three parts: a ``$co$'' subsystem (controllable and observable), the aforementioned ``$h$'' subsystem, and a ``$\bar{c}\bar{o}$'' subsystem that is completely isolated from the input-output channels. The dimensions of the ``$co$'',  ``$\bar{c}\bar{o}$'', and ``$h$'' subsystems are $2n_1$, $2n_2$, $2n_3$, respectively, where $n_1,n_2,n_3 \geq 0$ and $n_1+n_2+n_3 = n$. 
The system matrices in this canonical form are given by~\cite[Eq.~(67)]{ZGPG18}
\begin{equation*}
\bar{A} = \begin{bmatrix}
A_h^{11} & A_h^{12} & A_{12} & A_{13} \\
0 & A_h^{22} & 0 & 0 \\
0 & A_{21} & A_{co} & 0 \\
0 & A_{31} & A_{\bar{c}o} & A_{\bar{c}\bar{o}}
\end{bmatrix},\quad
\bar{B} = \begin{bmatrix}
B_h \\ 0 \\ B_{co} \\ 0
\end{bmatrix},\quad
\bar{C} = \begin{bmatrix}
0 & C_h & C_{co} & 0
\end{bmatrix},
\end{equation*}
with $\bar{D}=I_{2m}$. The state vector is partitioned as $\bar{\mathbf{x}}=[\mathbf{q}_h^{\top},\mathbf{p}_h^{\top},\mathbf{x}_{co}^{\top},\mathbf{x}_{\bar{c}\bar{o}}^{\top}]^{\top}$, where $\mathbf{q}_h$ and $\mathbf{p}_h$ are conjugate variables constituting the ``$h$'' subsystem, $\mathbf{x}_{co}$ corresponds to the ``$co$'' subsystem, and $\mathbf{x}_{\bar{c}\bar{o}}$ forms the isolated ``$\bar{c}\bar{o}$'' part.

As shown in \cite[Lem. 3.2]{ZPL20}, the system Hamiltonian $\mathbf{H}$ in the quantum Kalman canonical form can be written as $\mathbf{H} = \frac{1}{2}\bar{\mathbf{x}}^{\top}\tilde{\mathbb{H}}\bar{\mathbf{x}}$, where $\tilde{\mathbb{H}}$ is a real symmetric matrix, and the coupling operator is given by $\mathbf{L} = \Gamma\bar{\mathbf{x}}$ with a complex matrix $\Gamma\in\mathbf{C}^{m\times 2n}$. 
In the Kalman canonical form, the matrices $\tilde{\mathbb{H}}$ and $\Gamma$ inherit the block structure of the decomposition. 
Specifically, $\Gamma$ is of the form~\cite[Eq.~(13)]{ZPL20}
\begin{equation*}
\begin{bmatrix} 
\Gamma \\
\Gamma^\#
\end{bmatrix}= \begin{bmatrix}
0 & \Gamma_h & \Gamma_{co} & 0
\end{bmatrix},
\end{equation*}
while $\tilde{\mathbb{H}}$ has the block structure given in \cite[Eq.~(12)]{ZPL20}. The state variables satisfy the canonical commutation relations $[\bar{\mathbf{x}}(t),\bar{\mathbf{x}}(t)^{\top}] = \imath\bar{\mathbb{J}}_n$, where 
\begin{equation*}
\bar{\mathbb{J}}_n = \begin{bmatrix}
\mathbb{J}_{n_3} & 0 & 0 \\
0 & \mathbb{J}_{n_1} & 0 \\
0 & 0 & \mathbb{J}_{n_2}
\end{bmatrix}.
\end{equation*}

Under the quantum Kalman canonical form, the condition $[\mathbf{L},\mathbf{H}] = 0$ is equivalent to
\begin{equation*}
\Gamma\bar{\mathbb{J}}_n\tilde{\mathbb{H}}=0.
\end{equation*}
Thus, by \cite[Lemma 3.2]{ZPL20} we have
\begin{equation*}
\Gamma_{co}\bar{\mathbb{J}}_n\tilde{\mathbb{H}}=
\left[\begin{array}{c}
\Gamma \\
\Gamma^\# 
\end{array}\right]\bar{\mathbb{J}}_n\tilde{\mathbb{H}}=0,    
\end{equation*}
where
\begin{equation*}
\Gamma_{co}=\left[\begin{array}{cc}
\Gamma_{co,q} & \Gamma_{co,p} \\
\Gamma_{co,q}^\# & \Gamma_{co,p}^\#
\end{array}\right],
\end{equation*}
which yields 
\begin{equation*}
C_hA_h^{22}+C_{co}\mathbb{J}_{n_1}A_{12}^\top
+\frac{1}{2}C_{co}B_{co}\mathbb{J}_mB_h^\top=0,
\end{equation*}
and
\begin{equation}\label{Jan13-1}
C_{co}A_{co}=\frac{1}{2}C_{co}B_{co}C_{co},
\end{equation}
where
\begin{equation*}
C_{co}=V_m\Gamma_{co}=\left[\begin{array}{c}
C_{co,q} \\
C_{co,p}
\end{array}\right]=\sqrt{2}\left[\begin{array}{cc}
{\rm Re}(\Gamma_{co,q}) & {\rm Re}(\Gamma_{co,p}) \\
{\rm Im}(\Gamma_{co,q}) & {\rm Im}(\Gamma_{co,p})
\end{array}\right].    
\end{equation*}

The following theorem presents the characterization of BAE measurements by
the linear quantum system in the Kalman canonical form.

\begin{theorem}
The linear quantum system in the Kalman canonical form under the condition of $[\mbf{L},\mbf{H}]=0$ realizes the BAE measurements of $\mbf{q}_{\rm out}$ with respect to $\mbf{p}_{\rm in}$, i.e., the transfer function $\Xi_{\mbf{p}_{\rm in}\rightarrow \mbf{q}_{\rm out}}(s) \equiv 0$ 
if and only if
\begin{equation*}
C_{co,q}B_{co,p}=0,    
\end{equation*}
or equivalently,
\begin{equation*}
{\rm Re}(\Gamma_{co,q}) {\rm Re}(\Gamma_{co,p}^\top) =
{\rm Re}(\Gamma_{co,p})  {\rm Re}(\Gamma_{co,q}^\top).  
\end{equation*}
Meanwhile, the linear quantum system in the Kalman canonical form under the condition of $[\mbf{L},\mbf{H}]=0$ realizes the BAE measurements of $\mbf{p}_{\rm out}$ with respect to $\mbf{q}_{\rm in}$, i.e., the transfer function $\Xi_{\mbf{q}_{\rm in}\rightarrow \mbf{p}_{\rm out}}(s) \equiv 0$ if and only if 
\begin{equation*}
C_{co,p}B_{co,q}=0,    
\end{equation*}
or equivalently,
\begin{equation*}
{\rm Im}(\Gamma_{co,q}) {\rm Im}(\Gamma_{co,p}^\top) =
{\rm Im}(\Gamma_{co,p})  {\rm Im}(\Gamma_{co,q}^\top).  
\end{equation*}
Furthermore, the realization of both the two cases of BAE measurements imply that ${\rm Re}(\Gamma_{co,q}\Gamma_{co,p}^\top)$ is symmetric.
\end{theorem}

\begin{proof}
Partition the system matrices $C_{co}$ and $B_{co}$ as
\begin{equation*}
C_{co}=\left[\begin{array}{c}
C_{co,q} \\
C_{co,p}
\end{array}\right], ~~~~ B_{co}=\left[\begin{array}{cc}
B_{co,q} & B_{co,p} 
\end{array}\right],   
\end{equation*}
respectively. Under the condition $[\mbf{L},\mbf{H}]=0$, by Eq. \eqref{Jan13-1} we have
\begin{equation*}\begin{aligned}
C_{co}A_{co}^kB_{co}=\frac{1}{2^k}(C_{co}B_{co})^{k+1}.
\end{aligned}\end{equation*}

Thus, if $C_{co,q}B_{co,p}=0$, then $C_{co}A_{co}^kB_{co}$ is block lower triangular, which implies that $C_{co,q}A_{co}^kB_{co,p}=0$. By \cite[Theorem 4.1]{ZPL20}, the quantum Kalman canonical form realizes the BAE measurements of $\mbf{q}_{\rm out}$ with respect to $\mbf{p}_{\rm in}$. On the other hand, if the BAE measurements of $\mbf{q}_{\rm out}$ with respect to $\mbf{p}_{\rm in}$ is realized, then by \cite[Theorem 4.1]{ZPL20}
\begin{equation*}
C_{co,q}A_{co}^kB_{co,p}=0, ~~ k=0,1,\ldots,
\end{equation*}
which implies that $C_{co,q}B_{co,p}=0$. The proof of the BAE measurements of $\mbf{p}_{\rm out}$ with respect to $\mbf{q}_{\rm in}$ follows in a similar way. 
\end{proof}

\begin{example}
In \cite[Fig. 1(A)]{LOW+21}, it can be calculated that
\begin{equation*}
C_{co}=\left[\begin{array}{c}
C_{co,q} \\
C_{co,p} 
\end{array}\right]=\left[\begin{array}{cc}
\sqrt{\kappa} & 0 \\
0 & \sqrt{\kappa}
\end{array}\right],  ~~  
B_{co}=\left[\begin{array}{cc}
B_{co,q} & B_{co,p} 
\end{array}\right]=-\left[\begin{array}{cc}
\sqrt{\kappa} & 0 \\
0 & \sqrt{\kappa}
\end{array}\right],
\end{equation*}
which imply that
\begin{equation*}
\left\{\begin{array}{c}
C_{co,q}B_{co,p}=0, \\
C_{co,p}B_{co,q}=0.
\end{array}\right.     
\end{equation*}
Thus, both the BAE measurements of $\mbf{q}_{\rm out}$ with respect to $\mbf{p}_{\rm in}$ and $\mbf{p}_{\rm out}$ with respect to $\mbf{q}_{\rm in}$ are realized in this experiment. 
\end{example}

\bmrk
According to  \cite[Lemma 3.2]{ZPL20}, if $\tilde{\mathbb{H}}$ is of the form of  \cite[Eq. (12)]{ZPL20} and $\Gamma $ is of the form of \cite[Eq. (13)]{ZPL20}, then the linear system is of the Kalman canonical form. If $\Gamma_h \neq 0$, then the variables $\mbf{p}_h$ exist, which means there are QND variables.
\emrk

\subsection{Characterization of QND variables under QND interaction}

The QND interaction condition $[\mathbf{L},\mathbf{H}]=0$ analyzed in the preceding subsections not only enables BAE measurements but also singles out specific system observables that remain immune to measurement back‑action.  
The following theorem characterizes these QND variables.

\begin{theorem}\label{thm:QND_char}
Assume that $[\mathbf{L},\mathbf{H}]=0$ and consider the coupling operator 
$\mathbf{L} = C_-\mathbf{a} + C_+\mathbf{a}^\#$.
\begin{enumerate}
  \item If $C_-=-C_+$ ($\mathbf{p}$‑coupling) and $\Omega_-=-\Omega_+$, 
        then $\mathbf{p}$ is a QND variable provided the subsystem 
        $(\mathrm{Im}(\Omega_-),0,-\mathrm{Im}(C_-))$ or 
        $(\mathrm{Im}(\Omega_-),0,\mathrm{Re}(C_-))$ is observable.
        Moreover, BAE measurements are realized from $\mathbf{p}_{\mathrm{in}}$ 
        to $\mathbf{q}_{\mathrm{out}}$ and from $\mathbf{q}_{\mathrm{in}}$ 
        to $\mathbf{p}_{\mathrm{out}}$.
  \item If $C_- = C_+$ ($\mathbf{q}$‑coupling) and $\Omega_- = \Omega_+$, 
        then $\mathbf{q}$ is a QND variable provided the subsystem 
        $(\mathrm{Im}(\Omega_-),0,\mathrm{Im}(C_-))$ or 
        $(\mathrm{Im}(\Omega_-),0,\mathrm{Re}(C_-))$ is observable.
        Moreover, BAE measurements are realized from $\mathbf{p}_{\mathrm{in}}$ 
        to $\mathbf{q}_{\mathrm{out}}$ and from $\mathbf{q}_{\mathrm{in}}$ 
        to $\mathbf{p}_{\mathrm{out}}$.
\end{enumerate}
\end{theorem}

\begin{proof}
Recall from the quadrature representation that 
$\Lambda = \frac{1}{\sqrt{2}}[\,C_-+C_+ \;\; \imath(C_--C_+)\,] = [\,\Lambda_q \;\; \Lambda_p\,]$,
and that $\mathbb{B}$ and $\mathbb{C}$ are expressed in terms of $\Lambda_q$ and $\Lambda_p$.
For a variable to be a QND observable, it must belong to the uncontrollable yet observable
subspace, which requires either $\Lambda_q = 0$ or $\Lambda_p = 0$.

\noindent\textit{Case 1: $\Lambda_q = 0$, i.e.\ $C_- = -C_+$ ( $\mathbf{p}$‑coupling).}
In this case $C_--C_+ = 2C_- = -\sqrt{2}\,\imath\Lambda_p$ and $\Lambda_p = \sqrt{2}\,\imath\,C_-$.
A direct computation yields
\[
\mathbb{C} = 2\begin{bmatrix} 0 & -\mathrm{Im}(C_-) \\ 0 & \mathrm{Re}(C_-) \end{bmatrix},\qquad
\mathbb{B} = 2\begin{bmatrix} -\mathrm{Re}(C_-^\dagger) & \mathrm{Im}(C_-^\dagger) \\ 0 & 0 \end{bmatrix},
\]
and
\[
\mathbb{C}^\sharp\mathbb{C} = 4\begin{bmatrix} 0 & -\mathrm{Re}(C_-^\dagger)\mathrm{Im}(C_-)-\mathrm{Im}(C_-^\dagger)\mathrm{Re}(C_-) \\ 0 & 0 \end{bmatrix}.
\]
Imposing $\Omega_- = -\Omega_+$ gives $\Omega_--\Omega_+ = 2\Omega_-$, and the system matrix becomes
\[
\mathbb{A} = 2\begin{bmatrix} 0 & \mathrm{Re}(\Omega_-) \\ 0 & \mathrm{Im}(\Omega_-) \end{bmatrix}
           - 2\begin{bmatrix} 0 & -\mathrm{Re}(C_-^\dagger)\mathrm{Im}(C_-)-\mathrm{Im}(C_-^\dagger)\mathrm{Re}(C_-) \\ 0 & 0 \end{bmatrix}.
\]
The resulting equations of motion are
\begin{equation}\label{eq:QND_p}\begin{aligned}
\dot{\mathbf{q}} &= 2\bigl[\mathrm{Re}(\Omega_-) + \mathrm{Re}(C_-^\dagger)\mathrm{Im}(C_-) + \mathrm{Im}(C_-^\dagger)\mathrm{Re}(C_-)\bigr]\mathbf{p} \\
                 &\quad - 2\mathrm{Re}(C_-^\dagger)\mathbf{q}_{\mathrm{in}} + 2\mathrm{Im}(C_-^\dagger)\mathbf{p}_{\mathrm{in}}, \\
\dot{\mathbf{p}} &= 2\mathrm{Im}(\Omega_-)\mathbf{p}, \\
\mathbf{q}_{\mathrm{out}} &= -2\mathrm{Im}(C_-)\mathbf{p} + \mathbf{q}_{\mathrm{in}}, \\
\mathbf{p}_{\mathrm{out}} &= 2\mathrm{Re}(C_-)\mathbf{p} + \mathbf{p}_{\mathrm{in}}.
\end{aligned}\end{equation}
The second equation shows that $\mathbf{p}$ evolves independently of $\mathbf{q}$ and of the input quadratures that couple to $\mathbf{q}$. Hence $\mathbf{p}$ is a QND variable when the subsystem $(\mathrm{Im}(\Omega_-),0,-\mathrm{Im}(C_-))$ or
$(\mathrm{Im}(\Omega_-),0,\mathrm{Re}(C_-))$ is observable.  
The input–output relations in Eq. \eqref{eq:QND_p} directly imply the claimed BAE properties.

\noindent\textit{Case 2: $\Lambda_p = 0$, i.e.\ $C_- = C_+$ ( $\mathbf{q}$‑coupling).}
Assuming $\Omega_- = \Omega_+$, an analogous calculation yields the dynamical equations
\begin{equation}\label{eq:QND_q}
\begin{aligned}
\dot{\mathbf{q}} &= 2\mathrm{Im}(\Omega_-)\mathbf{q}, \\[2pt]
\dot{\mathbf{p}} &= -2\bigl[\mathrm{Re}(\Omega_-) + \mathrm{Im}(C_-^\dagger)\mathrm{Re}(C_-) + \mathrm{Re}(C_-^\dagger)\mathrm{Im}(C_-)\bigr]\mathbf{q} \\
                 &\quad - 2\mathrm{Im}(C_-^\dagger)\mathbf{q}_{\mathrm{in}} - 2\mathrm{Re}(C_-^\dagger)\mathbf{p}_{\mathrm{in}}, \\[2pt]
\mathbf{q}_{\mathrm{out}} &= 2\mathrm{Re}(C_-)\mathbf{q} + \mathbf{q}_{\mathrm{in}}, \\[2pt]
\mathbf{p}_{\mathrm{out}} &= 2\mathrm{Im}(C_-)\mathbf{q} + \mathbf{p}_{\mathrm{in}}.
\end{aligned}
\end{equation}
The first equation shows that $\mathbf{q}$ evolves independently of $\mathbf{p}$ and of the input quadratures that couple to $\mathbf{p}$. Hence $\mathbf{q}$ is a QND variable when the subsystem $(\mathrm{Im}(\Omega_-),0,\mathrm{Im}(C_-))$ or
$(\mathrm{Im}(\Omega_-),0,\mathrm{Re}(C_-))$ is observable.  
The input–output relations directly imply the stated BAE measurements.
\end{proof}

\begin{remark}
If $C_+ = 0$ and $C_-$ is real, the system matrices simplify to
\[
\mathbb{B} = -\begin{bmatrix} C_-^\dagger & 0 \\ 0 & C_-^\dagger \end{bmatrix},\qquad
\mathbb{C} = \begin{bmatrix} C_- & 0 \\ 0 & C_- \end{bmatrix},
\]
and
\[
\mathbb{A} = \begin{bmatrix} \mathrm{Im}(\Omega_-+\Omega_+) & \mathrm{Re}(\Omega_--\Omega_+) \\ -\mathrm{Re}(\Omega_-+\Omega_+) & \mathrm{Im}(\Omega_--\Omega_+) \end{bmatrix}
           - \frac{1}{2}\begin{bmatrix} C_-^\dagger C_- & 0 \\ 0 & C_-^\dagger C_- \end{bmatrix}.
\]
Even when $\Omega_- = \Omega_+$, which makes $\mathbb{A}$ block lower‑triangular,
the dynamics of $\mathbf{q}$ and $\mathbf{p}$ remain coupled through the off‑diagonal
block of $\mathbb{A}$.  Hence, while BAE measurement is feasible in this setting,
neither $\mathbf{q}$ nor $\mathbf{p}$ qualifies as a QND variable.
\end{remark}

\section{BAE measurements and QND variables synthesis}\label{synthesis}

\subsection{Realization of quantum BAE measurements via coherent feedback control}

As discussed in Section \ref{sec:qndbae}, when the system Hamiltonian matrix $\Omega$ is purely imaginary and the coupling operator matrix $\mathcal{C}$ is real or purely imaginary, quantum BAE measurements can be realized. However, if the sufficient conditions in Theorems \ref{lem:BAE1}-\ref{lem:BAE2} are not applied to the original system, we can cascade a beamsplitter and thus form a coherent feedback control network to realize quantum BAE measurements. Similar optical system structures of the coherent feedback control for squeezing enhancement \cite[Fig. 1(b)]{IYY+12} and noise reduction \cite[Fig. 1]{WZO2026} have been demonstrated experimentally. In this section, the original system $\boldsymbol{G}$, consisting of $n$ harmonic oscillators, $\mbf{a}=\left[\begin{array}{cccc}
\mbf{a}_1 & \mbf{a}_2 & \cdots & \mbf{a}_n
\end{array}\right]^\top$, is driven by $m=m_1+m_2$ input channels, whose $m_1$ input channels are described by $u_1$ in Fig. \ref{fig:system}, and the other $m_2$ input channels are represented by $u_2$. The system $\boldsymbol{G}$ is characterized by the following $(S,\mbf{L},\mbf{H})$ parameters
\begin{equation*}\begin{aligned}
S_G=\left[\begin{array}{cc}
S_{11} & S_{12} \\
S_{21} & S_{22}
\end{array}\right], \ \mbf{L}_G=\left[\begin{array}{c}
\mbf{L}_1 \\
\mbf{L}_2 
\end{array}\right], \ \Omega_G=\Delta(\Omega_-,\Omega_+),    
\end{aligned}\end{equation*}
where the coupling operators of the $m_1$ and $m_2$ channels are 
\begin{equation*}\begin{aligned}
\mbf{L}_1=k_{11}\mbf{a}+k_{12}\mbf{a}^\#, \\
\mbf{L}_2=k_{21}\mbf{a}+k_{22}\mbf{a}^\#,
\end{aligned}\end{equation*}
with the adjustable coupling strengths $k_{11}$, $k_{12}\in \mathbf{C}^{m_1 \times n}$, $k_{21}$, $k_{22}\in \mathbf{C}^{m_2 \times n}$.

\begin{figure}
    \centering
    \includegraphics[width=0.8\linewidth]{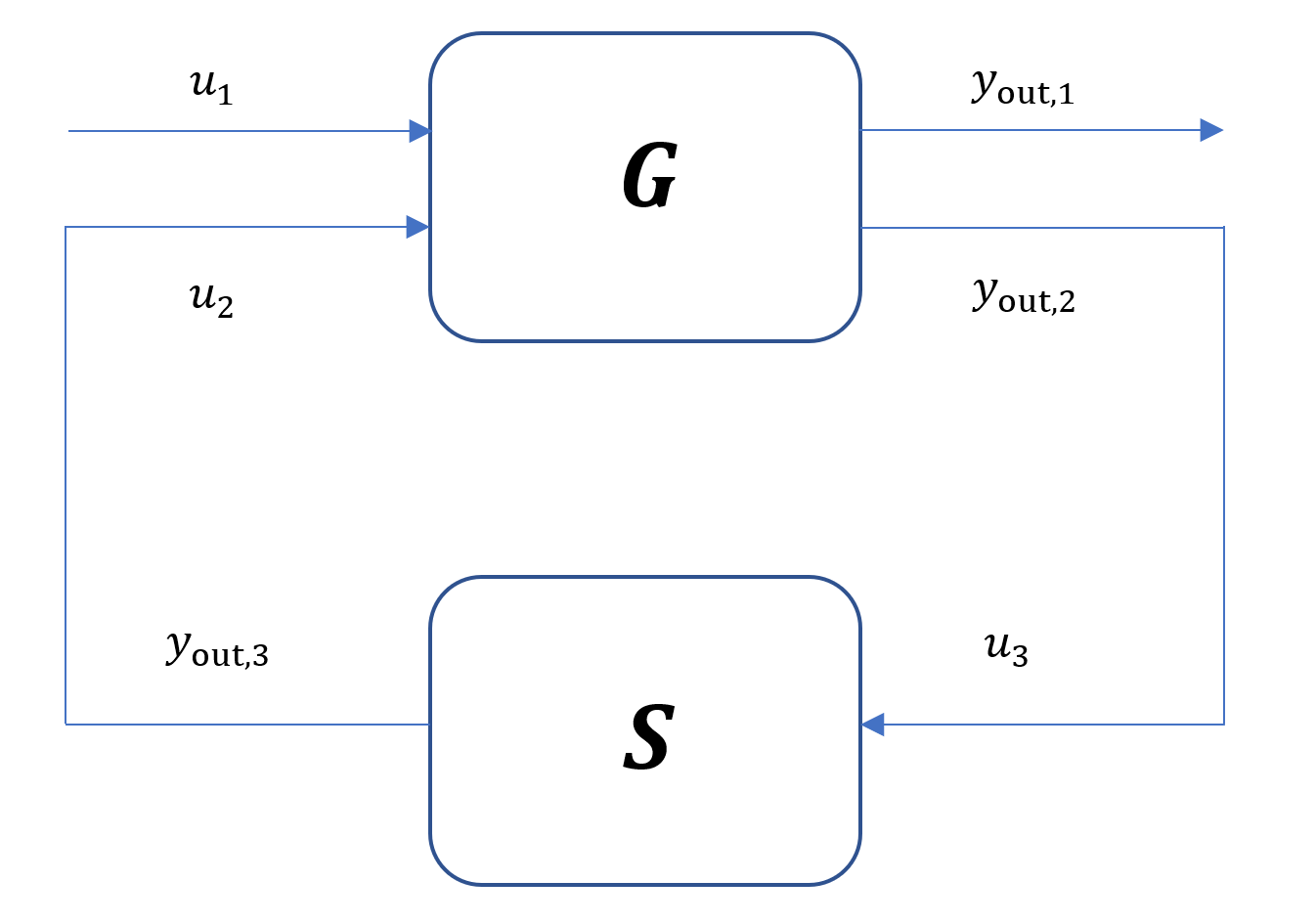}
    \caption{Coherent feedback network, where the original system’s dynamics are effectively modified by feeding back a portion of its output through a beamsplitter, thereby engineering the desired purely imaginary Hamiltonian matrix.}
    \label{fig:system}
\end{figure}

In contrast to the type-2 coherent feedback control scheme proposed in \cite[Fig. 9]{NY14}, we remove the controller and only use a beamsplitter $\boldsymbol{S}$ to form a coherent feedback network, see Fig. \ref{fig:system}. The feedback control is performed by modulating the corresponding $m_2$ output channels, i.e., $y_{\rm out,2}$, which is also the input $u_3$ of the beamsplitter $\boldsymbol{S}$, and the output $y_{\rm out,3}$ of the beamsplitter $\boldsymbol{S}$ is the input $u_2$ of the system $\boldsymbol{G}$. The following result shows that, if there exitsts a suitable choice of the coupling parameters, the coherent feedback network depicted in Fig. \ref{fig:system} can be engineered to achieve bilateral BAE measurements.




\begin{theorem}\label{thm:BAE}
If there exist complex coupling strength parameters $k_{ij}$, $i,j=1,2$, such that the overall coupling matrix $\bar{\mathcal{C}}$ is real or purely imaginary, and the overall Hamiltonian matrix $\bar{\Omega}$ is purely imaginary, then bilateral BAE measurements in the quantum coherent feedback network can be realized.  
\end{theorem}

\begin{proof}
The total $(S,\mathbf{L},\mathbf{H})$ parameters of the coherent feedback network in Fig. \ref{fig:system} can be calculated as
\begin{equation*}\begin{aligned}
S^{\rm red}&=S_{11}+S_{12}S_b(I-S_{22}S_b)^{-1} S_{21}, \\
\mbf{L}^{\rm red}&=\left[\begin{array}{cc}
\bar{C}_- & \bar{C}_+  
\end{array}\right]\left[\begin{array}{c}
\mbf{a} \\
\mbf{a}^\#      
\end{array}\right],
\end{aligned}\end{equation*}
where $\bar{C}_-=k_{11}+S_{12}S_b(I-S_{22}S_b)^{-1}k_{21}$, $\bar{C}_+=k_{12}+S_{12}S_b(I-S_{22}S_b)^{-1}k_{22}$, and the overall system Hamiltonian matrix 
\begin{equation*}
\bar{\Omega}=\Delta(\bar{\Omega}_-,\bar{\Omega}_+),    
\end{equation*}
where 
\begin{equation*}\begin{aligned}
\bar{\Omega}_-=\Omega_--\imath(k_{11}^\dagger S_b k_{21}-k_{21}^\dagger S_b^\dagger k_{11}), \\ 
\bar{\Omega}_+=\Omega_+-\imath(k_{11}^\dagger S_b k_{22}-k_{21}^\dagger S_b^\dagger k_{12}).
\end{aligned}\end{equation*}

The feedback network therefore reduces to an open linear quantum system with effective parameters $S^{\rm red}$, $\bar{\mathcal{C}}=\Delta(\bar{C}_-,\bar{C}_+)$, and $\bar{\Omega}$. If coupling strengths $k_{ij}$ can be chosen such that $\bar{\Omega}$ is purely imaginary and $\bar{\mathcal{C}}$ is either real or purely imaginary, then the coherent feedback network satisfies the hypotheses of Theorems \ref{lem:BAE1}--\ref{lem:BAE2}, and bilateral BAE measurements follow immediately.
\end{proof}

\begin{remark}
In Theorem \ref{thm:BAE}, it is assumed the plant parameters $k_{ij}$ are tunable, this is called quantum re-engineering; see for example \cite[Ex. 4.4]{JG10}.
\end{remark}

The following numerical example validates Theorem \ref{thm:BAE}.

\begin{example}
Assume that the original system $\boldsymbol{G}$ is with the following system Hamiltonian parameters
\begin{equation*}
\Omega_-=\left[\begin{array}{cc}
2 & 3+2\imath \\
3-2\imath & 4
\end{array}\right], \ \Omega_+=\left[\begin{array}{cc}
2 & 3-\imath \\ 
3-\imath & 5
\end{array}\right].    
\end{equation*}    
Clearly, the system Hamiltonian matrix $\Omega$ does not satisfy the sufficient condition given in Proposition \ref{prop:BAE}. The coupling strengths of the upper $m_1$ channels in Fig. \ref{fig:system} are set as follows
\begin{equation*}\begin{aligned}
k_{11}=\left[\begin{array}{cc}
1 & 1+\imath 
\end{array}\right], \ \ k_{12}=\left[\begin{array}{cc}
1 & 2-\imath 
\end{array}\right],
\end{aligned}\end{equation*}
while the coupling strengths of the lower $m_2$ channels
\begin{equation*}\begin{aligned}
k_{21}=\left[\begin{array}{cc}
1+\imath & 1+\imath 
\end{array}\right], \ \ k_{22}=\left[\begin{array}{cc}
1+\imath & 2+2\imath 
\end{array}\right].
\end{aligned}\end{equation*}
Set the scattering matrix $S=-\imath I$, then the overall Hamiltonian matrix $\bar{\Omega}$ of the coherent feedback network can be calculated as
\begin{equation*}
\bar{\Omega}_-=\left[\begin{array}{cc}
0 & -\imath \\
\imath & 0
\end{array}\right], \ \bar{\Omega}_+=\left[\begin{array}{cc}
0 & \imath \\ 
\imath & 3\imath
\end{array}\right],    
\end{equation*}
which is purely imaginary. Moreover, the overall coupling 
\begin{equation*}
\left[\begin{array}{cc}
\bar{C}_- & \bar{C}_+ 
\end{array}\right]=\left[\begin{array}{cccc}
1 & 2 & 1 & 1
\end{array}\right].    
\end{equation*}
By Theorem \ref{lem:BAE2}, bilateral quantum BAE measurements of $\mbf{q}_{\rm out}$ with respect to $\mbf{q}_{\rm in}$ and $\mbf{p}_{\rm out}$ with respect to $\mbf{p}_{\rm in}$ in Fig. \ref{fig:system} are realized. Since the system is controllable, no QND variable is present.
\end{example}

\subsection{Realization of QND variables via direct coupling}

In contrast to the two coherent feedback control design schemes proposed in \cite{NY14}, we introduce an additional auxiliary mode as a controller, as shown in Fig. \ref{fig:system1}. This open-loop control design scheme indirectly measures information about the original system modes by adjusting the direct interaction Hamiltonian $H_{\rm int}$ between the controller and the original system. It is theoretically proven that a linear combination of the original system modes can be constructed as a QND observable.

\begin{figure}
    \centering
\includegraphics[width=0.8\linewidth]{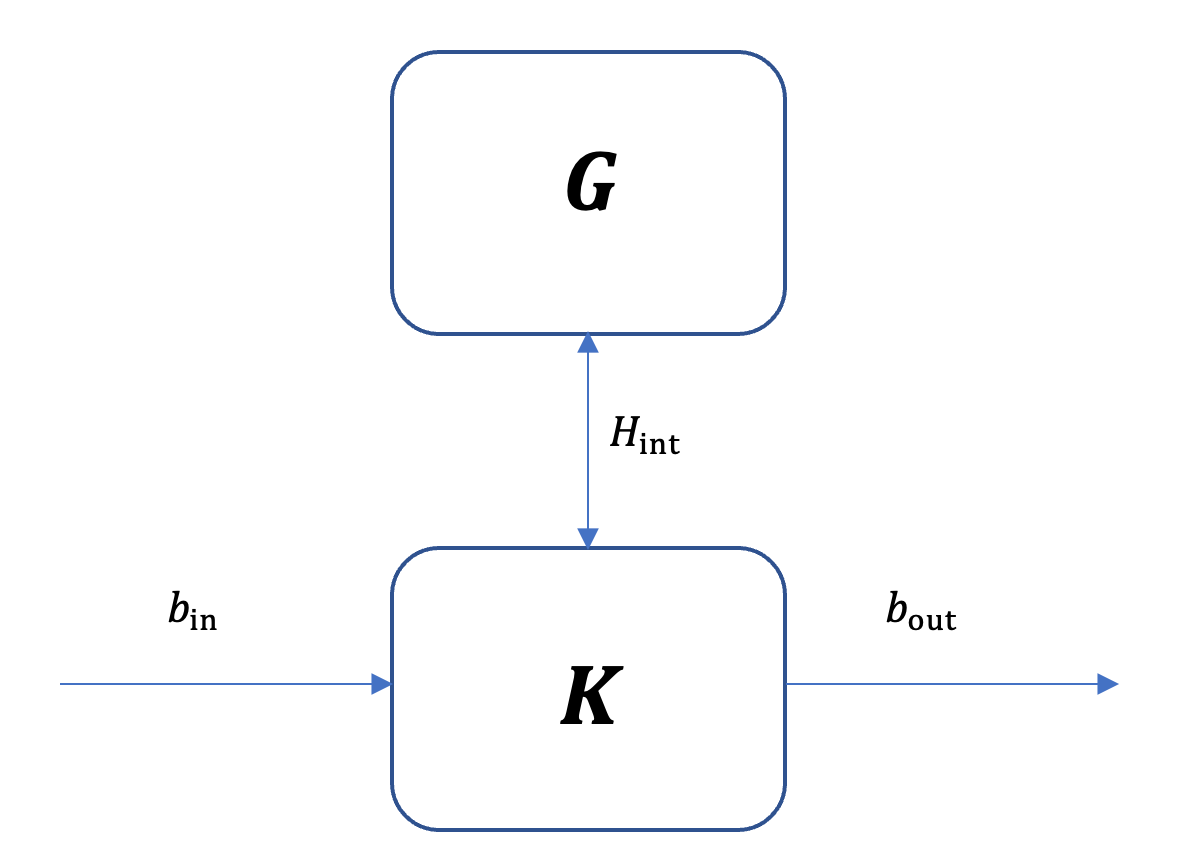}
    \caption{Quantum optomechanical system \cite{ZJ11,ZJ12,DCZF16}.}
    \label{fig:system1}
\end{figure}

The quantum optomechanical system in Fig. \ref{fig:system1} contains two cavity modes $\mbf{a}_1$, $\mbf{a}_2$ (the original system) and a mechanical oscillator mode $\mbf{a}_3$ (the controller). Under the rotating wave approximation, the total system Hamiltonian is given by \cite[Eq. (80)]{ZGPG18}
\begin{equation*}
\mbf{H}=-\frac{\Delta_1}{2}(\mbf{q}_1^2+\mbf{p}_1^2)-\frac{\Delta_2}{2}(\mbf{q}_2^2+\mbf{p}_2^2)+\frac{\omega_m}{2}(\mbf{q}_3^2+\mbf{p}_3^2)+(\lambda_1\mbf{q_1}+\lambda_2\mbf{q}_2)\mbf{q}_3,    
\end{equation*}
where the last term represents the interaction Hamiltonian $H_{\rm int}$ between cavity and mechanical oscillator. $\Delta_1$, $\Delta_2$ are the optical detunings for cavity modes $\mbf{a}_1$ and $\mbf{a}_2$, respectively. $\omega_m$ denotes the resonant frequency of the mechanical
oscillator. The coupling strengths between cavity and mechanical
oscillator are $\lambda_1$, $\lambda_2>0$. The coupling operator of the optomechanical system is $\mbf{L}=\sqrt{\kappa}\mbf{a}_3$, where $\kappa>0$ is the mechanical damping rate. According to the quantum Langevin equation, the system evolution equations can be calculated as
\begin{equation}\label{QNDevo}\begin{aligned}
\dot{\mbf{q}}_1&=-\Delta_1 \mbf{p}_1, \\
\dot{\mbf{p}}_1&=\Delta_1\mbf{q}_1-\lambda_1\mbf{q}_3, \\
\dot{\mbf{q}}_2&=-\Delta_2\mbf{p}_2, \\
\dot{\mbf{p}}_2&=\Delta_2\mbf{q}_2-\lambda_2\mbf{q}_3, \\
\dot{\mbf{q}}_3&=-\frac{\kappa}{2}\mbf{q}_3+\omega_m\mbf{p}_3-\sqrt{\kappa}\mbf{q}_{\rm in}, \\
\dot{\mbf{p}}_3&=-\omega_m\mbf{q}_3-\frac{\kappa}{2}\mbf{p}_3-\lambda_1\mbf{q}_1-\lambda_2\mbf{q}_2-\sqrt{\kappa}\mbf{p}_{\rm in},
\end{aligned}\end{equation}
and the input-output of the optomechanical system is
\begin{equation*}
\left[\begin{array}{c}
\mbf{q}_{\rm out} \\
\mbf{p}_{\rm out}
\end{array}\right]=\sqrt{\kappa}\left[\begin{array}{c}
\mbf{q}_3 \\
\mbf{p}_3
\end{array}\right]+\left[\begin{array}{c}
\mbf{q}_{\rm in} \\
\mbf{p}_{\rm in}
\end{array}\right].    
\end{equation*}

It can be verified that the uncontrollability of $\lambda_1\mathbf{q}_1+\lambda_2\mathbf{q}_2$ holds not only in the phase-shift regime $\Delta_1=\Delta_2=0$ \cite[Ex. 5.1(Case 3)]{ZGPG18} but also under the fixed parameter condition $\Delta_1=-\Delta_2\neq0$ and $\lambda_1=\lambda_2$. In the latter case, using the controllability matrix $\mathcal{C}_u\triangleq[B\;AB\;\cdots\;A^{n-1}B]$, one finds $c^\top\mathcal{C}_u=0$ with $c=[\lambda_1,0,\lambda_2,0,0,0]^\top$. That is, the linear combination $\lambda_1\mathbf{q}_1+\lambda_2\mathbf{q}_2$  lies in the orthogonal complement of the controllable subspace and is therefore uncontrollable. On the other hand, $\lambda_1\mathbf{q}_1+\lambda_2\mathbf{q}_2$ can be observed by the last two equations in Eq. \eqref{QNDevo}. Therefore, it is uncontrollable yet observable, which can be regarded as a QND variable. This is consistent with the QND variable given by the Kalman decomposition form derived in \cite[Ex. 5.1(Case 3)]{ZGPG18}. Moreover, when $\Delta_1=\Delta_2=0$ as stated in that phase-shift regime, either the cavity mode $\mbf{q}_1$ or $\mbf{q}_2$ can serve as a QND observable by tuning the coupling strength parameters $\lambda_1$ and $\lambda_2$ in Eq. \eqref{QNDevo}.

\begin{remark}
The parameter condition ($\Delta_1=-\Delta_2\neq0$, $\lambda_1=\lambda_2$) makes the two original system modes respond oppositely to the controller, so that the symmetric combination $\lambda_1\mathbf{q}_1+\lambda_2\mathbf{q}_2$ is completely decoupled from the input field. Consequently, this combination is uncontrollable yet remains observable from the output, which is exactly the required property for a QND variable that involves only the original system modes. Compared with the coherent feedback scheme proposed in \cite{NY14} which requires an auxiliary quantum controller and a feedback interconnection, the present approach using direct coupling (Fig.~2) is structurally simpler. Moreover, the resulting QND variable $\lambda_1\mathbf{q}_1+\lambda_2\mathbf{q}_2$ involves only the original system modes and does not contain any controller mode, facilitating experimental implementation and state estimation. 
\end{remark}

\section{Conclusion}\label{Conclu}

This paper has presented a systematic framework for realizing BAE measurements and QND variables in linear quantum systems from a control-engineering perspective. By analyzing the structural properties of the system parameters---the scattering matrix, coupling operator, and Hamiltonian---we have identified explicit conditions under which both bilateral and unilateral BAE measurements can be achieved. The analysis also reveals a fundamental connection between BAE measurements and QND interactions: the condition $[\mbf{L},\mbf{H}]=0$ not only guarantees that measurement back-action is evaded, but also singles out specific system observables that remain unperturbed by continuous monitoring.
For systems whose inherent parameters do not meet these conditions, we have shown that coherent feedback control provides an effective means to engineer the desired BAE performance, while QND variables can be realized through suitably designed direct coupling schemes. These results, developed consistently in annihilation-creation operator representation and further illuminated through the Kalman canonical form, provide a unified framework for the analysis and design of BAE measurements and QND variables in linear quantum systems. They offer clear guidelines for practical applications such as quantum sensing, gravitational wave detection, and quantum information processing. Future work may extend these techniques to nonlinear quantum systems and non-Markovian environments, further advancing the capabilities of quantum measurement and control.




\bmsubsection*{Funding}

This work is partially financially supported by Innovation Program for Quantum Science and Technology 2023ZD0300600, Guangdong Provincial Quantum Science Strategic Initiative No. GDZX2303007, Hong Kong Research Grant Council (RGC) under Grant No. 15213924, National Natural Science Foundation of China under Grants Nos. 62003111, 62473117, Natural Science Foundation of Guangdong Province under Grant No. 2025A1515060009, The Science Center Program of National Natural Science Foundation of China under Grant No. 62188101.

\bmsubsection*{Conflicts of Interest}

The authors declare no conflicts of interest.

\bmsubsection*{Data Availability Statement}

Data sharing is not applicable to this article as no datasets were generated or analyzed during the current study.

\bibliography{gzhang}



\end{document}